\documentclass[11pt,fleqn]{article}

\usepackage[text={6.75in,10in},centering,a4paper]{geometry}
\usepackage[colorlinks=true,linkcolor=blue,citecolor=blue]{hyperref}
\usepackage{latexsym}
\usepackage{amssymb}
\usepackage{amsfonts}
\usepackage{graphicx}
\usepackage{subfigure}
\usepackage{esint}

\def\eps{\varepsilon}

\def\L{{\cal L}}
\def\theoremknight{3.1 }
\def\exampleknight{3.5}

\newcommand{\qmbox}[1]{\quad\mbox{#1}\quad}

\def\phipin{\phi_{\rm pin}}
\def\Lpin{{\mathcal L}_{\rm pin}}
\def\L0{{\mathcal L}_0}
\def\bJ{{\bf J}}
\def\sech{{\rm sech}}
\def\eps{\varepsilon}
\def\iphi{\phi_{\rm in}}
\def\ophi{\phi_{\rm out}}
\def\ip{p_{\rm in}}
\def\op{p_{\rm out}}
\def\calL{{\mathcal L}}

\def\phimax{\phi_{\rm max}}

\def\mijnlim#1{\mathop{\rm lim}\limits_{#1}}

\newenvironment{arrayl}%
        {\renewcommand{\arraystretch}{1.5}
        \begin{array}{@{}r@{\hskip\arraycolsep}c@{\hskip\arraycolsep}l}}%
        {\end{array}\renewcommand{\arraystretch}{1}}
 
\newcommand{\beq}{\begin{equation}}
\newcommand{\eeq}{\end{equation}}
\newcommand{\eqref}[1]{(\ref{#1})}
\newtheorem{theorem}{Theorem}
\newtheorem{lemma}[theorem]{Lemma}
\newtheorem{corollary}[theorem]{Corollary}
\newenvironment{keywords}{\textbf{\small Keywords}\small:}{.}
\newenvironment{AMS}{\textbf{\small AMS subject %
    classifications}\small:}{.}
\newenvironment{proof}{%
        \begin{trivlist}\item[]{\bf Proof}\,\,}{%
        \hspace*{\fill} $\Box$\end{trivlist}}
\newenvironment{proofof}[1]{%
        \begin{trivlist}\item[]{\bf Proof of #1}\,\,}{%
        \hspace*{\fill} $\Box$\end{trivlist}}
\newtheorem{remark}[theorem]{Remark}

\title{\textbf{Pinned fluxons in a Josephson junction with a finite-length inhomogeneity}\thanks{G. Derks acknowledges a
    visitor grant of the Dutch funding agency NWO and the NWO-mathematics
    cluster NDNS${}^+$ and the hospitality of the CWI. }}

\author{
Gianne Derks\thanks{Department of Mathematics,
University of Surrey, Guildford, Surrey, GU2 7XH
({\tt g.derks@surrey.ac.uk})},
\and
Arjen Doelman\thanks{Mathematisch Instituut,
Leiden University,
P.O. Box 9512,
2300 RA Leiden,
the Netherlands
({\tt doelman@math.leidenuniv.nl})}
\and
Christopher J.K.\ Knight\thanks{Department of Mathematics,
University of Surrey, Guildford, Surrey, GU2 7XH
({\tt christopher.knight@surrey.ac.uk})},
\and
\and
Hadi Susanto\thanks{School of Mathematical Sciences, University of Nottingham, University Park, Nottingham, NG7 2RD ({\tt hadi.susanto@math.nottingham.ac.uk})}}

\begin{document}

\maketitle

\begin{abstract}
We consider a Josephson junction system installed with a finite length 
inhomogeneity, either of microresistor or of microresonator type. The 
system can be modelled by a sine-Gordon equation with a 
piecewise-constant function to represent the varying Josephson tunneling 
critical current. The existence of pinned fluxons depends on the length 
of the inhomogeneity, the variation in the Josephson tunneling critical 
current and the applied bias current. We establish that a system may 
either not be able to sustain a pinned fluxon, or -- for instance by 
varying the length of the inhomogeneity -- may exhibit various different 
types of pinned fluxons. Our stability analysis shows that changes of 
stability can only occur at critical points of the length of the 
inhomogeneity as a function of the (Hamiltonian) energy density inside 
the inhomogeneity -- a relation we determine explicitly. In combination 
with continuation arguments and Sturm-Liouville theory, we determine the 
stability of all constructed pinned fluxons. It follows that if a given 
system is able to sustain at least one pinned fluxon, there is exactly 
one stable pinned fluxon, i.e. the system selects one unique stable 
pinned configuration. Moreover, it is shown that both for microresistors 
and microresonators this stable pinned configuration may be 
non-monotonic -- something which is not possible in the homogeneous 
case. Finally, it is shown that results in the literature on localised 
inhomogeneities can be recovered as limits of our results on 
microresonators.
\end{abstract}




\section{Introduction}\label{sec.intro}

In this paper we consider a sine-Gordon-type equation
describing the gauge invariant phase difference of a long Josephson
junction
\begin{equation}
\phi_{tt} = \phi_{xx} - D\sin(\phi)+ \gamma - \alpha \phi_t,
\label{eq.junction}
\end{equation}
where $x$ and $t$ are the spatial and temporal variable respectively;
$\phi(x,t)$ is the Josephson phase difference of the junction;
$\alpha>0$ is the damping coefficient due to normal electron flow
across the junction; and $\gamma$ is the applied bias current. The
parameter $D$ represents the Josephson tunneling critical current,
which can  vary as a function of the spatial variable.

When $D$ is constant (without loss of generality, we can take $D=1$)
and there is no imposed current and dissipation, i.e.,
$\gamma=\alpha=0$, the system~(\ref{eq.junction}) is completely
integrable~\cite{ablo73} and has a family of travelling kink solutions
of the form
\begin{equation}
\phi(x,t) = \phi_0\left(\frac{x+vt+x_0}{\sqrt{1-v^2}}\right) ,
\qmbox{with} \phi_0(\xi) = 4\arctan(e^\xi) \qmbox{for any} |v|<1.
\label{fluxon}
\end{equation}
In the study of Josephson junctions, this kink represents a fluxon,
i.e.\ a magnetic field with one flux quantum
$\Phi_0\approx2.07\times10^{-15}$ Wb. If there is a small induced
current and dissipation but no inhomogeneity, then there is a unique
travelling fluxon whose wave speed in lowest order is given by $v =
\frac{\pi}{\sqrt{16(\alpha/\gamma)^2+\pi^2}}$ and no stationary
fluxons exist, see, e.g.,~\cite{ddvgv03}.

It was first suggested and shown in \cite{mcla78} that if the critical
current $D$ is locally perturbed, stationary fluxons can exist even if
an imposed current is present ($\gamma\neq 0$) and that a traveling
fluxon~(\ref{fluxon}) can be pinned by the inhomogeneity.  This
phenomenon is of interests from physical and fundamental point of view
because such an inhomogeneity could be present in experiments due to
the nonuniformity in the width of the transmission Josephson junction
line (see, e.g.,~\cite{akoh85,saka85}) or in the thickness of the
oxide barrier between the superconductors forming the junction (see,
e.g.,~\cite{serp87,vyst88}). About a decade after the first analysis
of this phenomenon, it is shown in~\cite{kivs91} that the interaction
between a soliton and an inhomogeneity can be non-trivial, i.e.\ an
attractive impurity, which is supposed to pin an incoming fluxon,
could totally reflect the soliton provided that there is no damping in
the system. Recently it is proven that the final state at which a
soliton exits a collision depends in a complicated fractal way on the
incoming velocity~\cite{good07}.

So far almost all of the analytical and theoretical work describes the
local inhomogeneity by a delta-like
function~\cite{good07,kivs89,kivs91,mcla78}. Yet, the length of an
inhomogeneity in real experiments is varying from 0.5$\lambda_J$
\cite{serp87} to 5$\lambda_J$ \cite{akoh85,saka85}, with $\lambda_J$
being the Josephson penetration depth. Therefore, such inhomogeneities
are not well described by delta-functions. Kivshar et
al.~\cite{kivshar88} have considered the \emph{time-dependent}
dynamics of a Josephson fluxon in the presence of this more realistic
setup, i.e.\ fluxon scatterings that take into account the finite size
of the defect, within the framework of a perturbation theory, i.e.,
when $\alpha,\,\gamma$ are small and $D\approx 1$. Piette and
Zakrzewski~\cite{piet07_2} recently studied the scattering of the
fluxon on a finite inhomogeneity, extending \cite{kivs91,good07} for
finite length defects in the case when neither applied bias current
nor dissipation is present.

In this paper, we consider the problem of a long Josephson junction
with a finite-length inhomogeneity and give a systematic analysis of
the existence and stability of stationary fluxons, i.e., the pinned
fluxons, as they lie at the heart of the interaction of the travelling
fluxons with the inhomogeneity. The Josephson tunneling critical
current (denoted by~$D$ in~(\ref{eq.junction})) is a function of space
and is modelled by the step-function
\begin{equation}
D(x;L,d) = \left\{
\begin{array}{lll}
d,&& |x|<L,\\
1,&& |x|>L.
\end{array}
\right.
\label{def}
\end{equation}
with $d\geq 0$. 
The inhomogeneity, as modelled by~\eqref{def}, can be fabricated
experimentally with a high controllability and precision, such that the
strength and the length of the defect $d$ and $2L$ can be made as one
wishes (see \cite{weid06,weid07} and references therein for reviews of
the experimental setups). 
When the parameter $d$ is greater or less than one, the inhomogeneity
is called a microresonator respectively microresistor. They can be
thought of as a locally thinned respectively thickened junction, which
provide less respectively more resistance for the Josephson
supercurrent to go across the junction barrier. Note that as
(\ref{eq.junction}) without inhomogeneity is translationally
invariant, it does not matter where the inhomogeneity is placed.
The existence and stability problem of pinned fluxons in \emph{finite}
Josephson junctions with inhomogeneity (\ref{def}) has been considered
numerically 
by Boyadjiev et al.\ \cite{andr06,boya06,boya07}. Here we consider an
infinitely long Josephson junction with inhomogeneity~(\ref{def}) and
provide a full analytical study of the existence and stability of
pinned fluxons, using dynamical systems techniques, Hamiltonian
systems ideas, and Sturm-Liouville theory.

For the existence of the pinned fluxons we observe that, as $D\equiv1$
for $|x|$ large, it follows immediately that the asymptotic fixed
points of~\eqref{eq.junction} are given by $\sin\phi=\gamma$, and the
temporally stable stationary uniform solutions are $\phi=\arcsin\gamma
+2k\pi$.  Bu definition, a pinned fluxon is a stationary solution
of~\eqref{eq.junction}, which connects~$\arcsin\gamma$
and~$\arcsin\gamma + 2\pi$. Hence a pinned fluxon is a solution of the
boundary value problem
\begin{equation}\label{eq.ham_ode_system}
\begin{array}{l}
  \phi_{xx} - D(x;L,d)\sin\phi +\gamma =0; \\[2mm]
  \displaystyle\lim_{x\to\infty}\phi(x)=\arcsin\gamma + 2\pi 
  \qmbox{and} \lim_{x\to-\infty} \phi(x)=\arcsin\gamma.
\end{array}
\end{equation}
First we observe that pinned fluxons can only exist for bounded values
of the applied bias current, $|\gamma|\leq 1$ (where this upperbound is
directly related to our choice to set $D\equiv 1$ outside the defect). 
Moreover, there are symmetries in this system. If
$\phi(x)$ is a pinned fluxon connecting $\arcsin\gamma$ (at
$x\to-\infty$) and $\arcsin\gamma+2\pi$ (at $x\to+\infty$), then
$\phi(-x)$ is a solution as well, connecting $\arcsin\gamma+2\pi$
($x\to-\infty$) and $\arcsin\gamma$ ($x\to+\infty$). So the second
solution is a pinned anti-fluxon. The symmetry implies that we can
focus on pinned fluxons and all results for pinned anti-fluxons follow
by using the symmetry $x\to -x$.
Another important symmetry is 
\[
\phi(x) \to 2\pi-\phi(-x) \qmbox{and} \gamma \to -\gamma.
\]
Thus if $\phi(x)$ is a pinned fluxon with bias current $\gamma$, then
$2\pi-\phi(-x)$ is a pinned fluxon with bias current $-\gamma$. This
means that we can restrict to a bias current $0\leq\gamma\leq 1$ and
the case $-1\leq\gamma<0$ follows from the symmetry above.

Furthermore, the differential equation in~\eqref{eq.ham_ode_system} is
a (non-autonomous) Hamiltonian ODE with Hamiltonian
\begin{equation}
  \label{eq.ham_ode}
H = \frac 12\,p^2 - D(x;L,d)(1-\cos\phi) +\gamma\phi, \qmbox{where}
p=\phi_x. 
\end{equation}
The non-autonomous term has the form of a step function, which implies
that on each individual interval $(-\infty, -L)$, $(-L,L)$, and
$(L,\infty)$ the Hamiltonian is fixed, though the value of the
Hamiltonian will vary from interval to interval. Therefore the
solutions of~\eqref{eq.ham_ode} can be found via a phase plane
analysis, consisting of combinations of the phase portraits for the
system with $D=1$ and $D=d$, see also~\cite{susa03} for a similar
approach to get existence of $\pi$-kinks.  In the phase plane
analysis, the length of the inhomogeneity~($2L$) is treated as a
parameter. For $x<-L$, the pinned fluxon follows one of the two
unstable manifolds of fixed point~$(\arcsin\gamma,0)$ of the reduced
ODE~\eqref{eq.ham_ode_system}.  Similarly, for $x>L$ the pinned fluxon
follows one of the stable manifolds of the fixed point
$(\arcsin\gamma+2\pi,0)$.  Finally, for $|x|<L$ the pinned fluxon
corresponds to a part of one of the orbits of the phase portrait for
the system with $D=d$.  The freedom in the choice of the orbit in this
system implies the existence of pinned fluxons for various lengths of
the inhomogeneity.  See Figure~\ref{fig.example_phase} for an example
of the construction of a pinned fluxons when $\gamma=0.15$ and
$d=0.2$.
\begin{figure}[htb]
  \centering
  \includegraphics[width=0.6\textwidth]{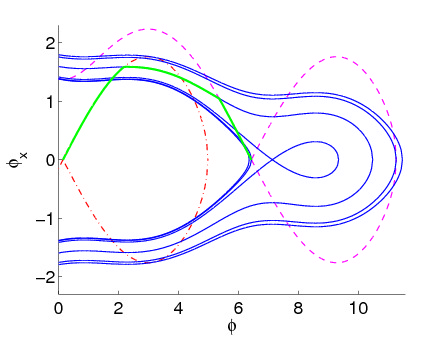}
  \caption{Phase portraits when $\gamma=0.15$ and $d=0.2$. The
    dash-dotted red curves are the unstable manifolds of
    $(\arcsin\gamma,0)$, the dashed magenta curves are the stable manifolds
    of $(2\pi+\arcsin\gamma,0)$, and the solid blue curves are examples of
    orbits for the dynamcis inside the inhomogeneity.  The bold green curve
    is an example of a pinned fluxon.}
  \label{fig.example_phase}
\end{figure}
 Orbits of a Hamiltonian
system can be characterised by the value of the Hamiltonian, hence
there is a relation between the value of the Hamiltonian inside the
inhomogeneity and the length of the inhomogeneity. The resulting
pinned fluxon is in $H^2(\mathbb{R})\cap C^1(\mathbb{R})$. As the
ODE~\eqref{eq.ham_ode_system} usually implies that the second
derivative of the pinned fluxon will be discontinuous, this is also
the best possible function space for the pinned fluxon solutions.

After analysing the existence of the pinned fluxons and having found a
plethora of possible pinned fluxons when a bias current is applied to
the Josephson junction (i.e., $\gamma\neq 0$), we will consider their
stability. First we will consider linear stability.  To derive the
linearised operator about a pinned fluxon~$\phipin(x;L,\gamma,d)$,
write $\phi(x,t)=\phipin(x;L,\gamma,d) + e^{\lambda
  t}v(x,t;L,\gamma,d)$ and linearise about~$v=0$ to get the eigenvalue
problem
\begin{equation}\label{eq.eigenval}
\Lpin v = \Lambda v, \qmbox{where} \Lambda=\lambda^2+\alpha\lambda,
\end{equation} 
and the linearisation operator $\Lpin(x;L,\gamma,d)$ is
\begin{equation}\label{eq.Lpin}
\Lpin(x;L,\gamma,d) = D_{xx} -D\cos \phipin (x;L,\gamma,d)=
  \left\{
  \begin{arrayl}
 D_{xx} - \cos\phipin(x;L,\gamma,d),&& |x|>L;\\    
 D_{xx}- d\,\cos\phipin(x;L,\gamma,d), && |x|<L.
  \end{arrayl}\right.
\end{equation}
The natural domain for $\Lpin$ is $H^2(\mathbb{R})$.  We call
$\Lambda$ an eigenvalue of $\Lpin$ if there is a function $v\in
H^2(\mathbb{R})$, which satisfies $\Lpin(x;L,\gamma,d) \, v =\Lambda
v$.  This operator is self-adjoint, hence all eigenvalues will be real.
Furthermore, it is a Sturm-Liouville operator, thus the Sobolev
Embedding Theorem gives that the eigenfunctions are continuously
differentiable functions in $H^2(\mathbb{R})$. Sturm's
Theorem~\cite{Titch} can be applied, leading to the fact that the
eigenvalues are simple and bounded from above.  Furthermore, if $v_1$
is an eigenfunction of $\Lpin$ with eigenvalue $\Lambda_1$ and $v_2$
is an eigenfunction of $\Lpin$ with eigenvalue $\Lambda_2$ with
$\Lambda_1>\Lambda_2$, then there is at least one zero of $v_2$
between any pair of zeros of $v_1$ (including the zeros at
$\pm\infty$). Hence, the eigenfunction $v_1$ has a fixed sign (no
zeros) if and only if $\Lambda_1$ is the largest eigenvalue of
$\Lpin$.
The continuous spectrum of $\Lpin$ is determined by the system
at $\pm\infty$. A short calculation shows that the continuous spectrum
is the interval $(-\infty,-\sqrt{1-\gamma^2})$.  

If the largest eigenvalue $\Lambda$ of $\Lpin$ is not positive or if
$\Lpin$ does not have any eigenvalues, then the pinned fluxon is
linearly stable, otherwise it is linearly unstable. This follows
immediately from analysing the quadratic
$\Lambda=\lambda^2+\alpha\lambda$. If $\Lambda\leq0$, then both
solutions $\lambda$ have non-positive real part. However, if
$\Lambda>0$ is then there is a solution $\lambda$ with positive real
part. Furthermore, the $\lambda$-values of the continuous spectrum
also have non-positive real part as the continuous spectrum of $\Lpin$
is on the negative real axis. 

The linear stability can be used to show nonlinear stability.  The
Josephson junction system without dissipation is Hamiltonian.  Define
$P=\phi_t$, $u=(\phi,P)$, then the equations~(\ref{eq.junction}) can
be written as a Hamiltonian dynamical system with dissipation on an
infinite dimensional vector space of $x$-dependent functions, which is
equivalent to $H^1(\mathbb{R})\times L^2(\mathbb{R})$:
\[
\frac{d}{dt} u =\bJ \,\delta {\mathcal H} (u) -\alpha \mathbf{D} u,
\qmbox{with}
\bJ = \left(
\begin{array}{cc}
  0&1\\-1&0
\end{array}\right), \quad
\mathbf{D} = \left(
\begin{array}{cc}
  0&0\\0&1
\end{array}\right),
\]
and
\begin{equation}\label{eq.pde_ham}
\begin{array}{lll} {\mathcal H}(u) &=& \frac
  12\displaystyle\int_{-\infty}^\infty \left[P^2 + \phi_x^2 +
    2\,D(x;L,d)\,(\sqrt{1-\gamma^2}-\cos\phi) \right] \, dx\\
&&{} -\gamma\displaystyle  \int_0^\infty[\phi-\arcsin\gamma-2\pi]\,dx
    + \gamma\int_{-\infty}^0[\phi-\arcsin\gamma]\, dx.
\end{array}
\end{equation}
Here we have chosen the constants terms in the $\gamma$-integrals such
that they are convergent for the fluxons. Furthermore, for any
solution $u(t)$ of~\eqref{eq.junction}, we have
\begin{equation}\label{eq.Ham_dyn}
\frac{d}{dt} \mathcal{H} (u) = -\alpha \int_{-\infty}^\infty P^2 dx
\leq 0.
\end{equation}
As a pinned fluxon is a stationary solution, we have
$D\mathcal{H}(\phipin,0)=0$ and the Hessian of $\mathcal{H}$ about a fluxon is
\[
D^2\mathcal{H}(\phipin,0) = 
\left(\begin{array}{cc}
  -\Lpin & 0\\0&I
\end{array}\right).
\]
If $\Lpin$ has only strictly negative eigenvalues, then it follows
immediately that $(\phipin,0)$ is a minimum of the Hamiltonian
and~\eqref{eq.Ham_dyn} gives that all solutions nearby the pinned
fluxon will stay nearby the pinned fluxon, see also~\cite{ddvgs07}.

After this introduction, we will start the paper with an overview of
simulations for the interaction of travelling fluxons and the
inhomogeneity in~\eqref{eq.junction} for various values of $d$, $L$,
$\gamma$ and $\alpha$. This will motivate the analysis of the
existence and stability of the pinned fluxons in the following
sections. We start the analysis of the existence and stability of
pinned fluxons by looking at a microresistor with $d=0$. The advantage
of the case $d=0$ is that several explicit expressions can be derived
and technical difficulties can be kept to a minimum, while it is also
representative of the general case $d<1$. It will be shown that for
$\gamma=0$ there is exactly one pinned fluxon for each length of the
inhomogeneity. For $\gamma>0$, a plethora of solutions starts
emerging. There is a minimum and maximum length outside which the
inhomogeneity cannot sustain pinned fluxons. Between the minimal and
the maximal length there are at least two pinned fluxons, often more.
At each length between the minimum and maximum, there is exactly one
stable pinned fluxon. If the length of the interval is (relatively)
large, the stable pinned fluxons are non-monotonic. Note that stable
non-monotonic fluxons are not possible in homogeneous systems, since
for a homogeneous system the derivative of the fluxon is an
eigenfunction for the eigenvalue zero of the operator associated with
the linearisation about the fluxon. If the fluxon is non-monotonous,
then this eigenfunction has zeros. As the linearisation operator is a
Sturm-Liouville operator, this implies that the operator must have a
positive eigenvalue as well, hence the non-monotonous fluxon is
unstable. However, for inhomogeneous systems, the derivative of the
fluxon is usually not differentiable, hence cannot give rise to an
eigenvalue zero (since the eigenfunctions have to be $C^1$) and stable
non-monotonic fluxons are in principle possible. This shows that the
inhomogeneity can give rise to qualitatively different fluxons.

For the existence analysis of the pinned fluxons, the length of the
inhomogeneity will be treated as a parameter. The pinned fluxons
satisfy an inhomogeneous Hamiltonian ODE whose Hamiltonian is constant
inside the inhomogeneity. It will be shown that the existence and type
of pinned fluxons can be parametrised by the value of this
Hamiltonian. The length of the inhomogeneity is determined by the
value of the Hamiltonian and the type of pinned fluxon, leading to curves
relating the length~$2L$ and the value of the Hamiltonian inside the
inhomogeneity. In~\cite{knight09} it is shown, in the general setting
of an inhomogeneous wave equation, that changes in stability of the
pinned fluxons can be associated with critical points of the length
function relating $L$ and the value of the Hamiltonian.  The results
of this paper together with Sturm-Liouville theory give the stability
properties of the pinned fluxons in the general setting.

After giving full details for the case $d=0$, for which the stability
issue can be settled in dependent of~\cite{knight09}, an overview of
the results for $d>0$ is given. The general microresistor case
($0<d<1$) is very similar to the case $d=0$. The microresonator case
($d>1$) has some different features, but the same techniques as before
can be used to analyse the existence and stability. We finish the
analysis of the microresonator case by looking at the special case
where microresonators approximate a localised inhomogeneity. We
explicitly look at microresonators with $d = \frac\mu{2L}$ and $L$
very small. For $\gamma$, $\alpha$ and $\mu$ small, the asymptotic
results from~\cite{mcla78} are recovered. Even in the limit of
localised inhomogeneities, our work generalises~\cite{mcla78}, since our
methods allows us to consider $\gamma$, $\alpha$ and $\mu$ larger
as well.

The paper concludes with some further observations, conclusions and
ideas for future research.

\section{Simulations}\label{sec.sim}

To put the analysis of the existence and stability of the pinned
fluxons in the next sections in a wider context, we look first at
simulations of the interaction of a travelling fluxon with an
inhomogeneity. Recall that in absence of dissipation and induced
currents ($\alpha=0=\gamma$), the system~\eqref{eq.junction} without
an inhomogeneity ($D\equiv 1)$, has a family of travelling fluxon
solutions~\eqref{fluxon}, for each wave speed $|v|<1$, while if there
is a small induced current and dissipation, but no inhomogeneity, then
there is a unique travelling fluxon.

First we look at the case $\alpha=0=\gamma$ (no induced current, no
dissipation) and the inhomogeneity of microresistor type with $d=0$.
If the length is too short, the fluxon will not be captured, but its
speed will be reduced by the passage through the inhomogeneity.  If
the length of the inhomogeneity is sufficiently large, the travelling
fluxon will be captured. Some radiation is released in this process
and the fluxon ``bounces'' backwards and forwards around the defect,
especially if the length is ``just long enough''. This is consistent
with the results in~\cite{piet07_2} where a detailed analysis of the
interaction of a fluxon with an inhomogeneity is studied in the case
that no induced current and dissipation are present.  An illustration
is given in Figure~\ref{fig.sim_d_0_g_0}.
\begin{figure}[htb]
\centering
\quad\hfill\includegraphics[width=.35\textwidth]{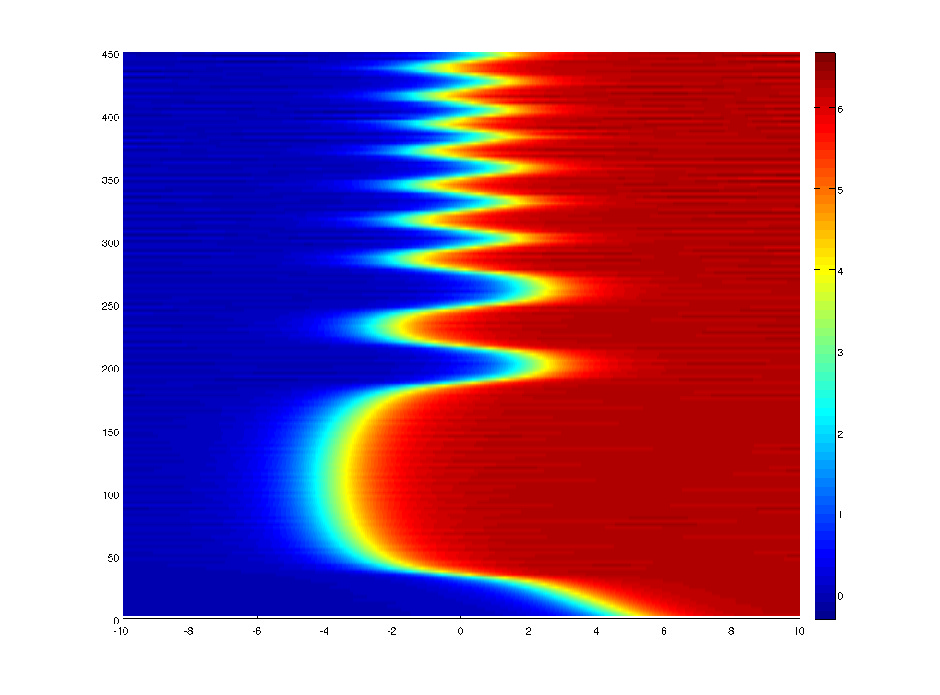}\hfill%
\includegraphics[width=.35\textwidth]{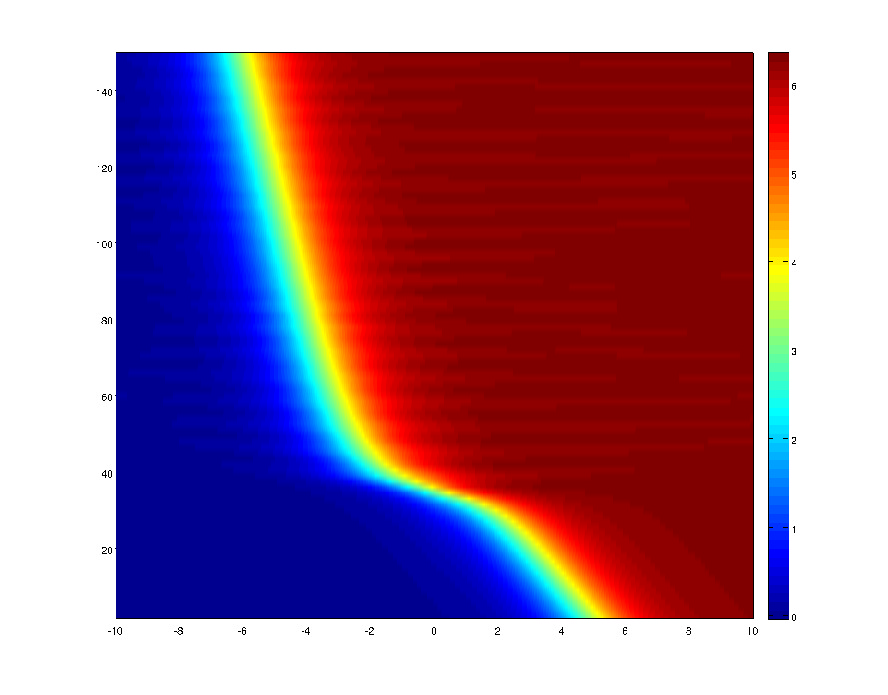}\hfill\quad
\caption{Simulation of a travelling wave with speed $v=0.1$
  approaching an inhomogeneity with $d=0$ when there is no induced
  current ($\gamma=0$) or dissipation ($\alpha=0$). The inhomogeneity
  is positioned in the middle (around the zero position). The length of
  the inhomogeneity on the left is $0.38$ and the travelling fluxon is
  captured by the inhomogeneity; note that the ``bounce'' of the
  fluxon is a lot larger than the length of the inhomogeneity. The
  length of the inhomogeneity on the right is $0.36$ and the pinned
  fluxon can just escape, but its speed is significantly reduced. }
  \label{fig.sim_d_0_g_0}\end{figure}
Note that the length of the defect which captures the fluxon is a lot
smaller than the initial amplitude of the ``bounce'' of the fluxon.
Observations suggest that the minimal length for the inhomogeneity to
capture the travelling fluxon increases if the wave speed increases.

Next we look at the system with a microresistor with $d=0$, now with
an induced current $\gamma=0.1$ and varying lengths and values of
$\alpha$.  We start again with an inhomogeneity of length $0.38$
($L=0.19$). When $\gamma=0$, this microresistor captures a fluxon with
speed $v=0.1$. With an induced current, it cannot capture a fluxon,
however large we make $\alpha$, i.e., however slow the fluxon
becomes. This is illustrated in Figure~\ref{fig.sim_d_0_g_0.1}.
\begin{figure}[htb]
\centering
\includegraphics[width=.33\textwidth]{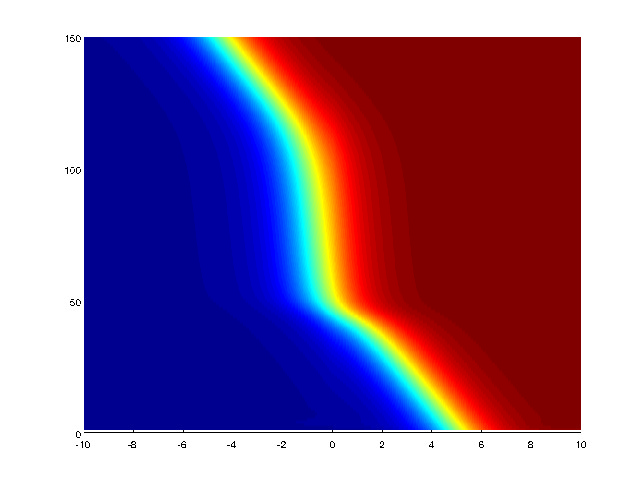}%
\includegraphics[width=.33\textwidth]{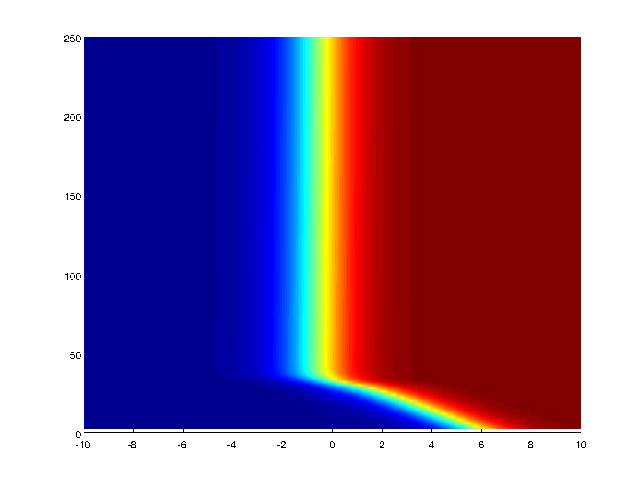}%
\includegraphics[width=.33\textwidth]{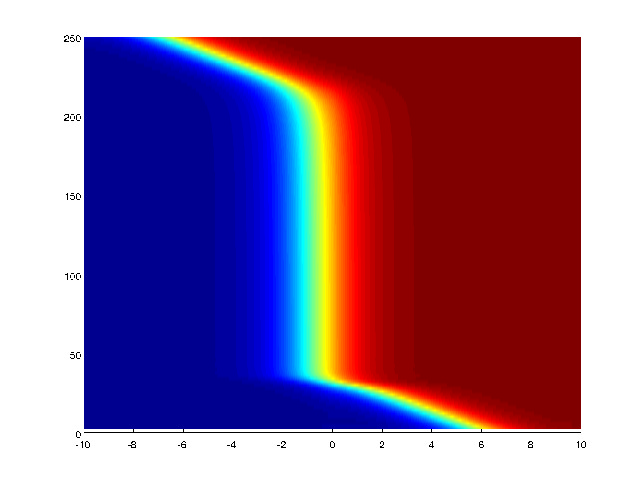}
\caption{Simulation of a travelling fluxon approaching an
  inhomogeneity with $d=0$ when the induced current is
  $\gamma=0.1$. On the left, the length is $0.38$. Here the
  dissipation is $\alpha=0.9$, but however large $\alpha$ is taken,
  the fluxon is never captured. In the middle and right plots, the
  length is $0.44$. In the middle the dissipation is $\alpha=0.48$ and
the fluxon is captured, on the right the dissipation is $\alpha=0.47$
and the fluxon can escape.}
\label{fig.sim_d_0_g_0.1}
\end{figure}
The microresistor slows the fluxon down for a while, but eventually
the fluxon escapes with the same speed as it had earlier (as this
speed is unique in a system with $\alpha,\,\gamma\neq 0$). The
simulations suggest that the smallest length which can capture a
fluxon is $0.44$ ($L=0.22$). In the next section, it will be shown
that for $\alpha,\,\gamma\neq 0$, there is a minimal length under
which no pinned fluxon can exist. This explains why the inhomogeneity
with the shortest length cannot capture even a very slow travelling
fluxon. In Figure~\ref{fig.sim_d_0_g_0.1} it is illustrated that, if
the length can sustain pinned fluxons, the capture depends on the
dissipation (hence on the speed of the incoming fluxon). If the
dissipation is sufficiently large, hence the speed sufficiently slow,
the pinned fluxon will be captured.

A longish defect in a microresistor will also capture the travelling
wave and the resulting pinned fluxon is not monotonic, see
Figure~\ref{fig.sim_d_0_g_0.1_long}! 
\begin{figure}[htb]
\centering
\includegraphics[width=.5\textwidth]{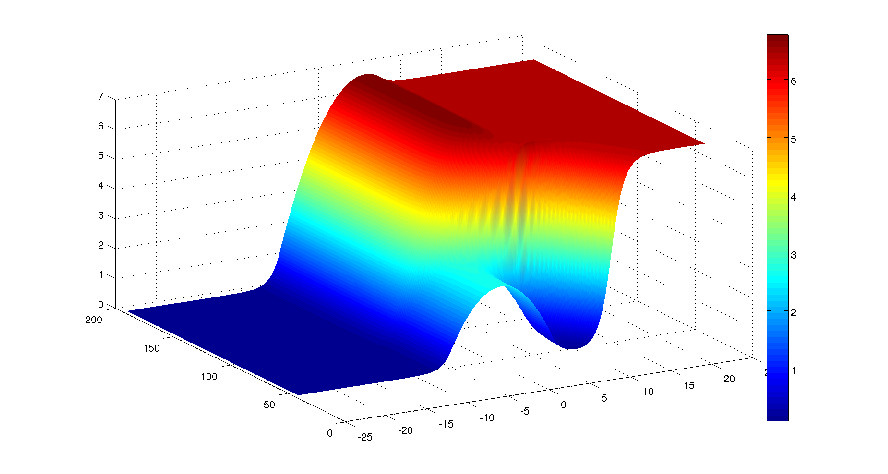}\hfill%
\includegraphics[width=.5\textwidth]{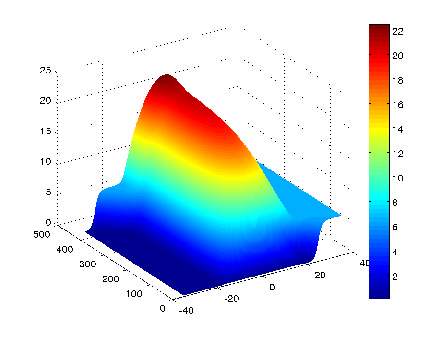}%
\caption{Simulation of a travelling fluxon approaching a longish
  inhomogeneity with $d=0$ when the induced current is $\gamma=0.1$
  and dissipation is $\alpha=0.5$. On the left, the length is $12.5$,
  the travelling wave is captured and a non-monotonic pinned fluxon is
  formed. On the right, the length is $35$ and the travelling wave
  escapes after a while, leaving in its wake a ``bump'' connecting
  $2\pi+\arcsin\gamma$ at both ends. Note that the vertical scale and
  coloring is different in both figures; as a reference point, the
  travelling wave on the right is the same in both cases.}
\label{fig.sim_d_0_g_0.1_long}
\end{figure}
The length of the inhomogeneity is substantial, so the stationary
shape connecting the far field rest states at $\arcsin\gamma$ is a
``bump''. This ``bump'' is present at all the rest states
$\arcsin\gamma+2k\pi$ for $\gamma\neq 0$ as $\arcsin\gamma+2k\pi$ is
not an equilibrium for the dynamics with $d\neq 1$. From a phase plane
analysis it can be seen that the amplitude of the homoclinic
connection to $\arcsin\gamma+2k\pi$ grows with the length~$L$ of the
defect. As shown in Figure~\ref{fig.sim_d_0_g_0.1_long}, for $L=6.25$, the travelling fluxon
travels into this ``bump'' and gets captured. The resulting pinned
fluxon is not monotonic. In the next section, the family of all
possible pinned fluxons is analysed and it is shown that for long
lengths the stable pinned fluxon is non-monotonic. Moreover, it
follows that there is an upper limit on the length of inhomogeneities
that can sustain pinned fluxons. This is illustrated on the right in
Figure~\ref{fig.sim_d_0_g_0.1_long}. The travelling fluxon seems to be
captured initially by the inhomogeneity, but after a while it escapes
again. However large the dissipation is taken, this will always
happen, illustrating that no pinned fluxons can exist.

Next we consider a microresonator with $d=2$. As before, we consider
the case without an induced current ($\gamma=0$) first. In this case, the
fluxon is never captured. For the smaller lengths the fluxon reflects,
for larger lengths the fluxon seems to get trapped, but it escapes
after a while. This is illustrated in Figure~\ref{fig.sim_d_2_g_0} for
a microresonator with length $0.1$. In the next section, it will be
will shown that a system with a microresonator and no induced current
has indeed no stable pinned fluxons. 
\begin{figure}[htb]
\centering
\quad\hfill
\includegraphics[width=.35\textwidth]{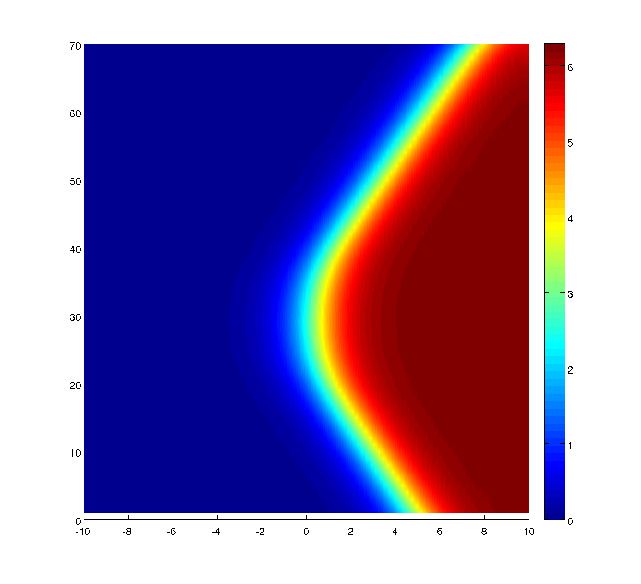}\hfill%
\includegraphics[width=.35\textwidth]{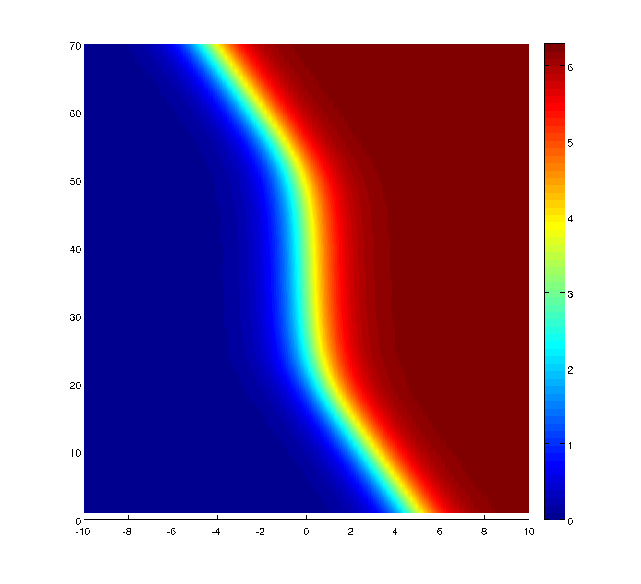}\hfill\quad
\caption{Simulation of a travelling wave approaching an inhomogeneity
  with $d=2$ and length $0.1$, when there is no induced current and no
  dissipation ($\gamma=0=\alpha$). The speed on the left is $0.21$ and
  the travelling wave is bounced by the inhomogeneity. The length on
  the right is $0.22$ and at first the pinned fluxon seems to be
  captured by the inhomogeneity, but after while it travels through
  the inhomogeneity and seems to resume its original speed. }
  \label{fig.sim_d_2_g_0}\end{figure}

After the induction-less system, we consider a system with a
microresonator with $d=2$ and an induced current $\gamma=0.1$. As with
the microresistor there is a minimum length, under which the
microresonator cannot capture a fluxon. The simulations suggest that
the minimum length is $0.42$ ($L=0.21$). In
Figure~\ref{fig.sim_d_2_g_0.1}, it is illustrated that a
microresonator with length $0.40$ cannot capture a fluxon with
$\alpha=0.9$, whilst a microresonator with length $0.42$ can capture a
fluxon with $\alpha=0.3$, but it cannot for $\alpha=0.29$.
\begin{figure}[htb]
\centering
\includegraphics[width=.33\textwidth]{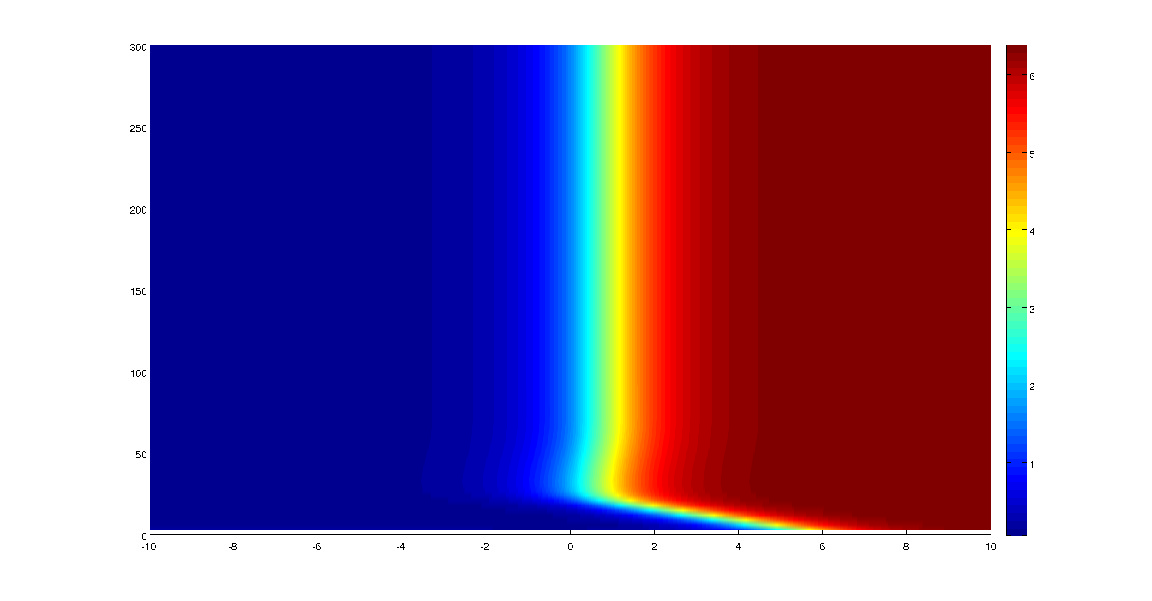}\hfill%
\includegraphics[width=.33\textwidth]{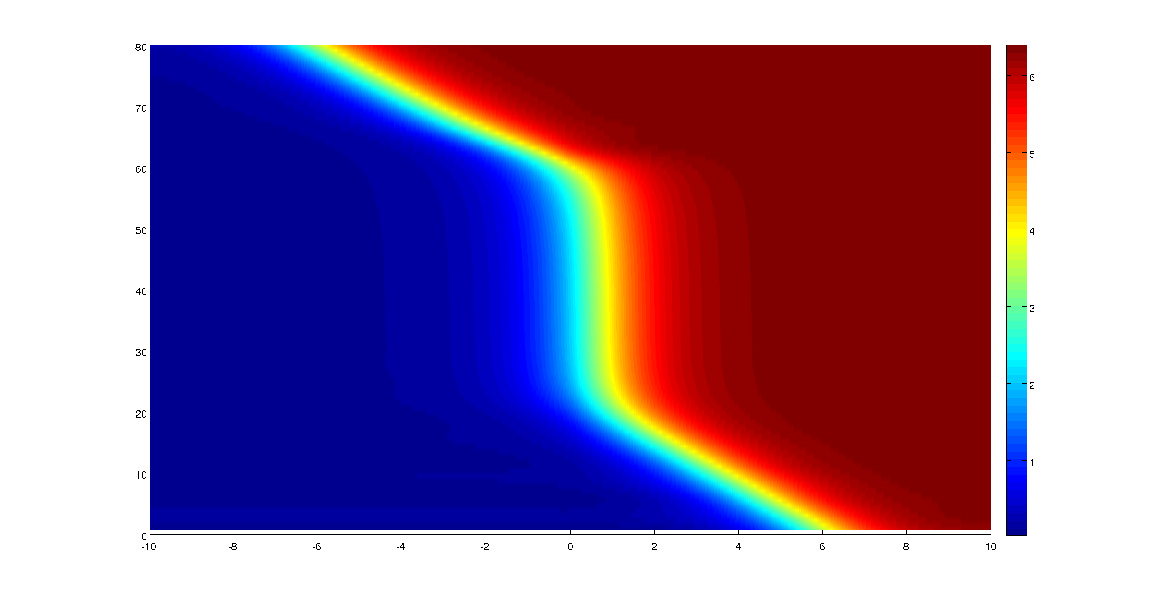}\hfill%
\includegraphics[width=.33\textwidth]{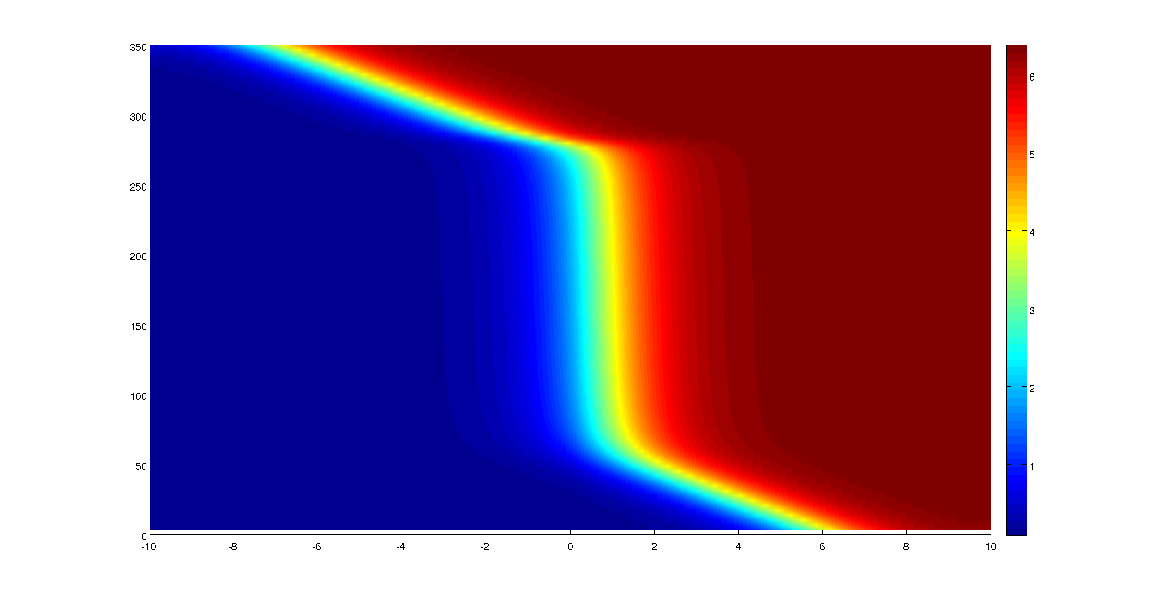}
\caption{Simulation of a travelling wave approaching an inhomogeneity
  with $d=2$ and length $0.1$, when there an induced current
  ($\gamma=0.1$). On the left and middle is a microresonator with
  length $0.42$. On the left the dissipation is $\alpha=0.3$ and the
  fluxon is captured, whilst in the middle the dissipation is
  $\alpha=0.29$ and the fluxon escapes. On the right, the length is
  $0.4$ and the dissipation is $\alpha=0.9$ and the fluxon still
  escapes as the length is too short for a pinned fluxon to exist. }
  \label{fig.sim_d_2_g_0.1}\end{figure}
This is consistent with the results in the next sections where it is
shown that for $\alpha,\,\gamma\neq 0$ there exists a minimal length
under which no pinned fluxons can be sustained by the inhomogeneity.
If the length can just sustain pinned fluxons, then there are both a
stable and an unstable pinned fluxon close to each other. In the left
panels of Figures~\ref{fig.sim_d_2_g_0.1} and~\ref{fig.sim_d_10_g_0.1}
it can be observed that initially the travelling fluxon approaches the
unstable pinned fluxon, but then reflects to the stable one and
settles down.

Finally we consider a microresonator with a longer length for which
the travelling fluxon gets captured and becomes a non-monotonic pinned
fluxon. In Figure~\ref{fig.sim_d_10_g_0.1}, it is illustrated that for
a microresonator with $d=10$, length~2 ($L=1$), the travelling fluxon
at $\gamma=0.1$ and $\alpha=0.2$ gets attracted to a non-monotonic
pinned fluxon. Note that for microresonators (i.e., $d>1$), the stable
non-monotonic pinned fluxons have a ``dip'' as opposed to the ones for the
microresistors which have a ``bump''.
\begin{figure}[htb]
\centering
\quad\hfill\includegraphics[width=.4\textwidth]{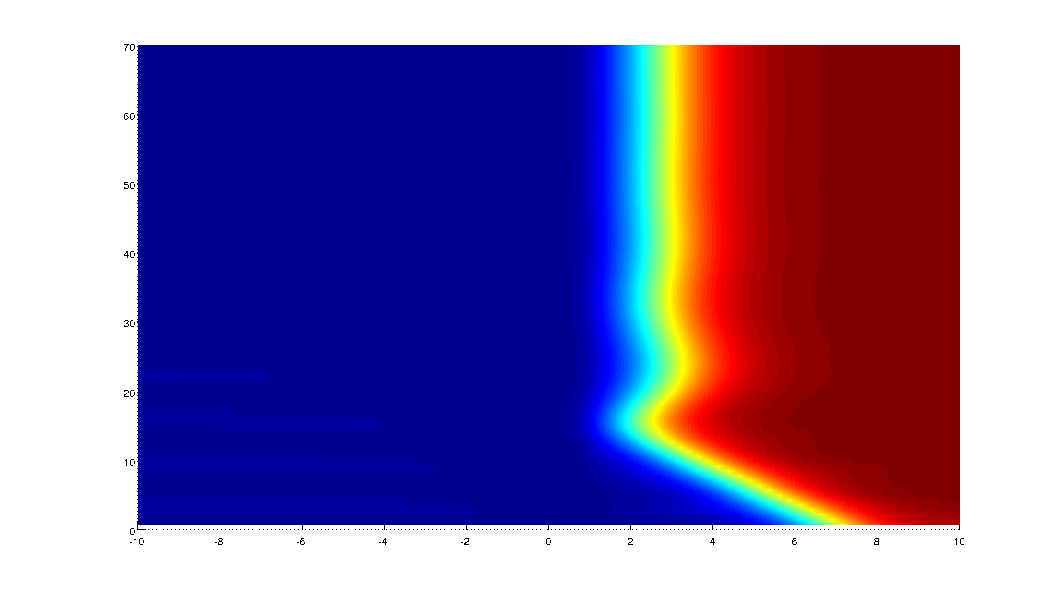}\hfill%
\includegraphics[width=.4\textwidth]{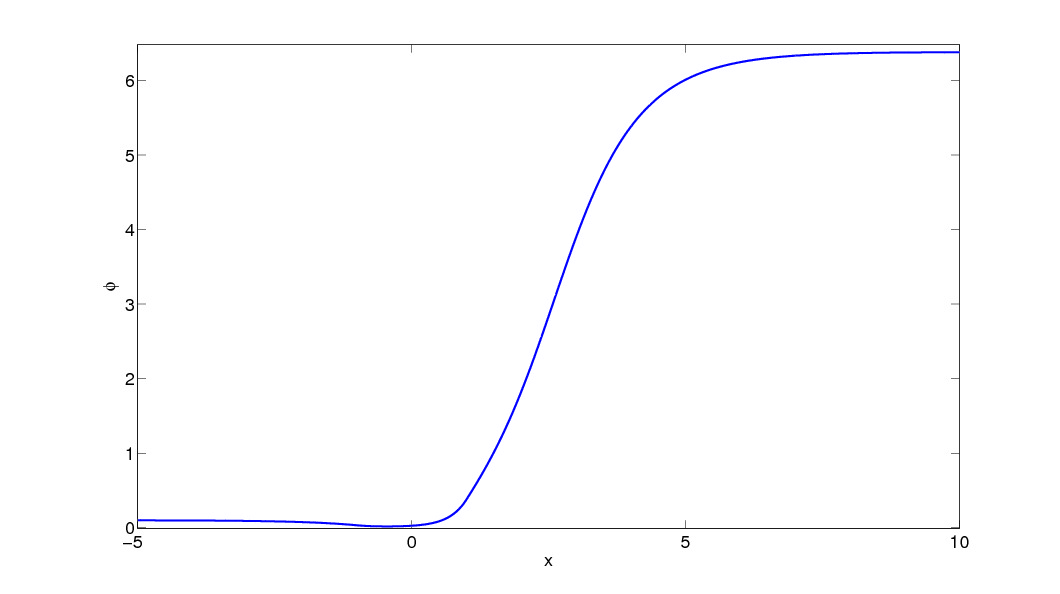}\hfill\quad%
\caption{Simulation of a travelling wave approaching an inhomogeneity
  with $d=10$ and length $2$, when the induced current is $\gamma=0.1$
  and the dissipation is $\alpha=0.2$. The resulting wave is
  non-monotonic as can be seen on the right. Due to the weaker
  dissipation, it takes some time for the wave to converge to its
  stable shape. Initially, the travelling wave approaches the
  monotonic unstable pinned fluxon, then deflects from it and
  converges to the non-monotic stable one.}
  \label{fig.sim_d_10_g_0.1}\end{figure}

\section{No resistance (d=0)}
\label{sec.d=0}

We now analyse the existence and stability of the pinned fluxons in a
microresistor and a microresonator.  First we consider the case when
there is no resistance in the inhomogeneity, hence a microresistor
with $d=0$. This case provides a good illustration of the richness of
the family of pinned fluxons, shows the essence of the analytic
techniques for the existence and stability analysis, and has less
technical complications than the more general values of $d$. The
existence analysis for the case with no bias current ($\gamma=0$) is
quite different from the case when a bias current is applied
($\gamma>0$). So we will consider them separately.

\subsection{Existence of pinned fluxons without applied bias current}
\label{sec.d=0,gamma=0}
For $\gamma=0$, the pinned fluxon has to connect the stationary states
at $\phi=0$ and $\phi=2\pi$. In the background dynamics of the
ODE~\eqref{eq.ham_ode_system} with $D\equiv1$, the unstable manifold
of $(0,0)$ coincides with the stable manifold of $(2\pi,0)$, as
follows immediately by analysing the Hamiltonian~\eqref{eq.ham_ode}
with $D\equiv 1$. These coinciding manifolds are denoted by a red
curve in the phase portrait sketched in
Figure~\ref{fig.phase_d_0_g_0}.
\begin{figure}[htb]
\centering
\includegraphics[width=.5\textwidth]{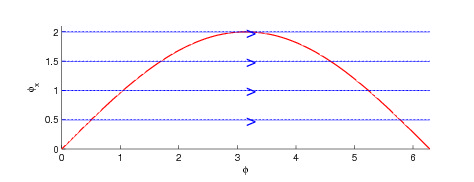}\hfill%
\includegraphics[width=.5\textwidth]{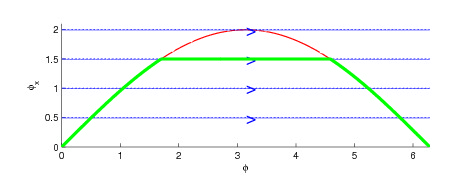}
\caption{Phase portraits  of the ODE~\eqref{eq.ham_ode_system} for
  $\gamma=0$ and $d=0$. The red curve represents the coinciding stable
  and unstable manifolds of the asymptotic fixed points. The blue
  curves are orbits for the system inside the inhomogeneity. In the
  sketch on the right, the green curve represents a pinned fluxon. }
  \label{fig.phase_d_0_g_0}
\end{figure}
This curve represents the unperturbed sine-Gordon
fluxon~\eqref{fluxon}.  The orbits generated by the Hamiltonian system
with $D\equiv0$ are straight lines. In Figure~\ref{fig.phase_d_0_g_0},
samples of these orbits are given by the blue lines. Any blue line
that crosses the red line can be used to form a pinned fluxon. An
example is given in the panel on the right in
Figure~\ref{fig.phase_d_0_g_0}, where the green curve represents a
pinned fluxon in $H^2(\mathbb{R})\cap C^1(\mathbb{R})$.

As can be seen from Figure~\ref{fig.phase_d_0_g_0}, the value of the
Hamiltonian inside the inhomogeneity is a convenient parameter to
characterise the pinned fluxons. The points of intersection for the
blue and red curves are denoted by $(\iphi,\ip)$ respectively
$(\ophi,\op)$ for the first respectively second intersection. It
follows immediately that $\ip=\op$ and $\ophi=2\pi-\iphi$.
Furthermore, the expression for the Hamiltonian,~\eqref{eq.ham_ode},
gives the following relations for $\iphi$ and $\ip$: $0 = \frac
12\,\ip^2 - (1-\cos\iphi)$ ($D\equiv1$) and $h = \frac 12\,\ip^2$
($D\equiv0$), with $0<h\leq 2$ where $h$ is the value of the
Hamiltonian inside the inhomogeneity.  Thus
\begin{equation}\label{eq.in_g=0}
\ip(h)=\sqrt{2h} \qmbox{and} \iphi(h) = \arccos(1-h), \qmbox{with}  0<h\leq 2.
\end{equation}
Inside the inhomogeneity ($|x|<L$), the pinned fluxon related to the
value~$h$ satisfies  
$h=\frac 12 \phi_x^2$, {thus} $\phi_x=\sqrt{2h}$.
Hence the half-length~$L$ and the parameter~$h$ are related by
\begin{equation}
\label{eq.length_d=0}
L=\int_{-L}^0 dx = \int_{\iphi(h)}^\pi \frac{d\phi}{\phi_x} = 
\int_{\iphi(h)}^\pi \frac{d\phi}{\sqrt{2h}} =
\frac{\pi-\arccos(1-h)}{\sqrt{2h}}.
\end{equation}
As the numerator is a monotonic decreasing function of~$h$ and the
denominator is monotonic increasing, it follows immediately that $L$
is a monotonic decreasing function of $h$. The function $L$ takes
values in $[0,\infty)$ as $\mijnlim{h\to0}L(h)=\infty$ and
$\mijnlim{h\to2}L(h)=0$. The $h$-$L$ plot is given in
Figure~\ref{fig.length_d=0,g=0}.  We summarise the existence
results for pinned fluxons without a bias current in the following
lemma.
\begin{figure}[htb]
  \centering
  \includegraphics[height=.25\textheight]{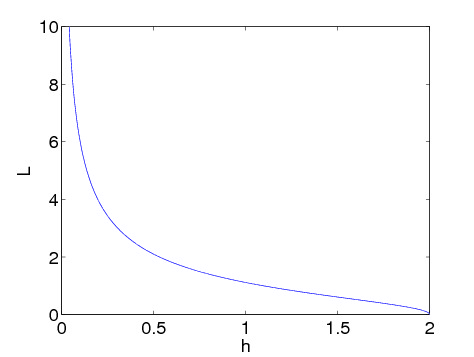}
  \caption{Plot of the length~$L$ as a function of~$h$, the value of
  the Hamiltonian in the inhomogeneity, for $\gamma=0$ and $d=0$.}
  \label{fig.length_d=0,g=0}
\end{figure}
\begin{lemma}
  \label{lem.exist_d=0,g=0}
  Let $\gamma=0$ and $d=0$. For any length~$2L$ of the inhomogeneity,
  there is a unique pinned fluxon, for which the Hamiltonian inside
  the inhomogeneity has the value~$h(L)$ as implicitly given
  by~\eqref{eq.length_d=0}. Define $x^*$ to be the shift such that
  $\phi_0(-L+x^*)=\iphi$ (see~\eqref{fluxon} for the definition of
  $\phi_0$), then the pinned fluxon is given explicitly by
\begin{equation}
\phipin (x;L,0,0)= \left\{ 
  \begin{array}{ll}\renewcommand{\arraystretch}{1.1}
  \phi_0(x+x^*),& x<-L,\\
  \pi+\frac{\pi-\arccos(1-h)}{L}\,x, &|x|<L,\\
  \phi_0(x-x^*),& x>L.
  \end{array}
\right.
\label{pf}
\end{equation}
\end{lemma}

\subsection{Existence of pinned fluxons with bias current}
\label{sec.d=0,gamma>0}

For $\gamma>0$, the pinned fluxon has to connect the stationary states
at $\phi=\arcsin \gamma$ and $\phi=2\pi+ \arcsin\gamma$. In the
background dynamics with $D\equiv1$ the unstable manifold of
$\phi=\arcsin\gamma$ coincides no longer with the stable manifold of
$2\pi+\arcsin\gamma$. Furthermore, the orbits of the dynamics
inside the inhomogeneity are parabolic curves instead of straight
lines.  These two changes add substantial richness to the family of
pinned fluxons.

Let us first consider the phase portraits. In
Figure~\ref{fig.phase_d=0_g=.15} we consider $\gamma=0.15$ as a
typical example to illustrate the ideas.  In the dynamics with
$D\equiv1$, the unstable manifolds to $\arcsin \gamma$ are denoted by
red curves, while the stable manifolds to $2\pi+\arcsin\gamma$ are
denoted by magenta curves. The larger $\gamma$ gets, the wider the gap
between the unstable and stable manifold becomes.
\begin{figure}[htb]
\centering
\includegraphics[width=0.55\textwidth]{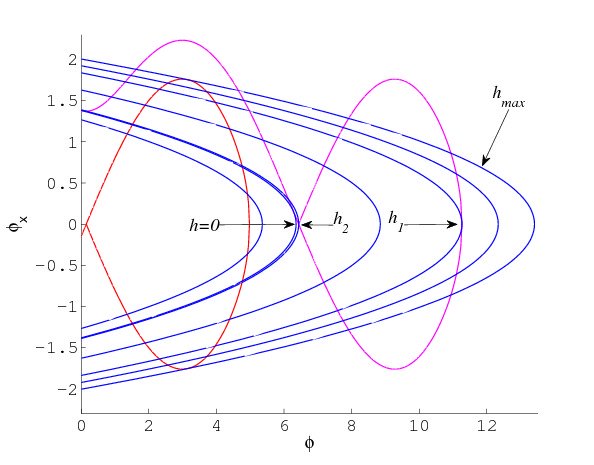}
  \includegraphics[width=0.44\textwidth]{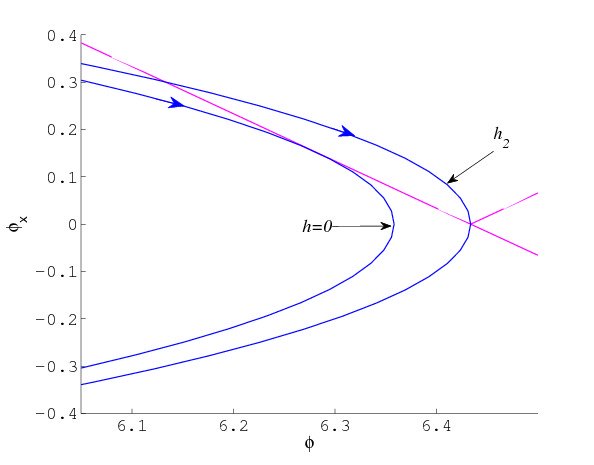}
  \caption{Phase portrait at $\gamma=0.15$ and $d=0$. On the
    right is a zoom into the area around $(\phi,\phi_x)=(2\pi,0)$.}
  \label{fig.phase_d=0_g=.15}
\end{figure}
The dynamics within the inhomogeneity with $D\equiv0$ are denoted by
blue curves. These blue curves are nested and can be parametrised with
a parameter~$h$, using the Hamiltonian~\eqref{eq.ham_ode} with
$D\equiv0$:
\[
\frac12(\phi_x)^2 +\gamma\phi = H_0(\gamma) +h,
\]
where $H_0(\gamma)$ is given by the value of the
Hamiltonian~\eqref{eq.ham_ode} on the magenta stable manifold ($D\equiv1$):
\begin{equation}
  \label{eq.H0}
  H_0(\gamma) = \gamma\arcsin\gamma -
(1-\sqrt{1-\gamma^2})+2\pi\gamma  .
\end{equation}
Thus the value of $h$ increases as the extremum of the blue curves is
more to the right.

For the existence of pinned fluxons, a blue curve has to connect the
red unstable manifold with the magenta stable manifold. In
Figure~\ref{fig.phase_d=0_g=.15}, the furthest left possible blue
curve for which pinned fluxons may exist, is the one indicated with
$h=0$. In the zoom on the right, it can be seen that this curve just
touches the magenta stable manifold.  The blue curve intersects the
red unstable manifold twice, both points give rise to a pinned fluxon,
as sketched in Figure~\ref{fig.phase_d=0_g=.15_8}. Obviously, the
pinned fluxon in the second plot in Figure~\ref{fig.phase_d=0_g=.15_8}
will occur in a defect with a shorter length than the one in the first
plot.
\begin{figure}[htbp]
    \centering
    \includegraphics[width=0.35\textwidth]{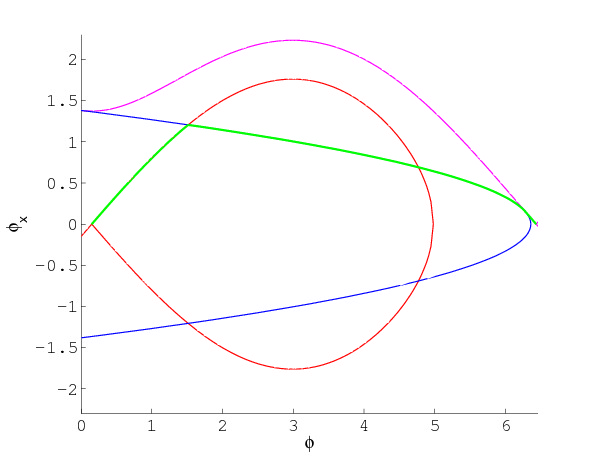}
    \hfill
\includegraphics[width=0.35\textwidth]{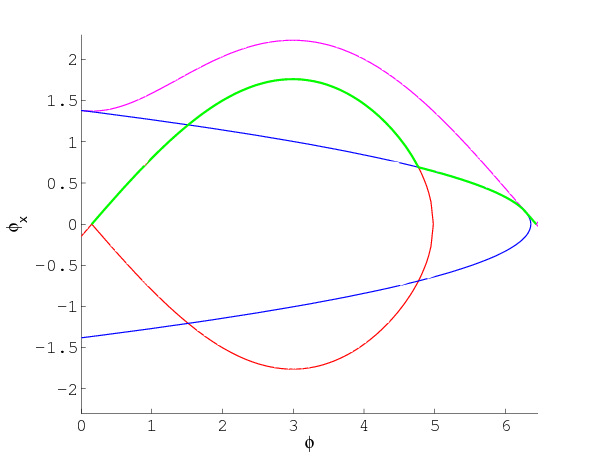}
    \hfill
\includegraphics[width=0.28\textwidth]{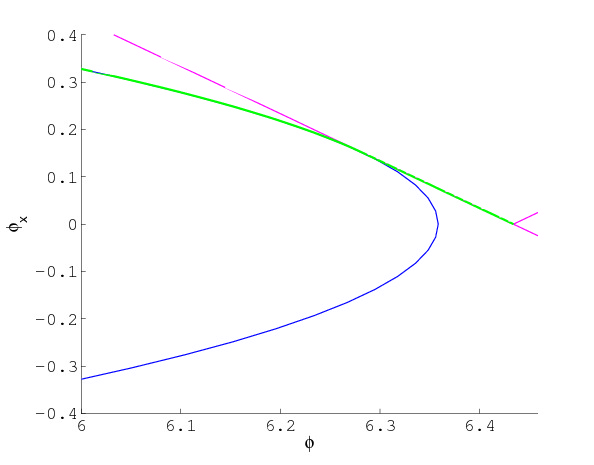}
  \caption{Phase portrait at $\gamma=0.15$ and $d=0$ with the furthest
    left blue curve for which pinned fluxons exist. There are two
    pinned fluxons possible, represented by green line. On the right
    is a zoom into the area around $(\phi,\phi_x)=(2\pi,0)$.}
  \label{fig.phase_d=0_g=.15_8}
\end{figure}
The furthest right possible curve that gives rise to pinned fluxons is
marked with~$h_{\rm max}$ in Figure~\ref{fig.phase_d=0_g=.15}. This
blue orbit touches the red unstable manifold and crosses the magenta
unstable manifolds in 5 points. All these points represent different
pinned fluxons, hence 5 pinned fluxons can be associated with this
curve. Moreover, for~$h$ just below~$h_{\rm max}$, the blue curve
intersects the red curve twice (while it still intersects the magenta
curve~5 times: there are~10 different pinned fluxons associate to such
value of~$h$.

In general, the pinned fluxons are determined by two points in the
phase plane: the point where pinned fluxon enters the inhomogeneity
(i.e. the crossing from the red unstable manifold to the blue orbit),
this point will be denoted by $(\iphi,\ip)$ and the point where the
pinned fluxon leaves the inhomogeneity (i.e. the crossing from the
blue orbit to the magenta stable manifold), this point will be denoted
by $(\ophi,\op)$. Thus the points $(\iphi,\ip)$ and $(\ophi,\op)$ are
determined by the set of equations
\begin{equation}\label{eq.qp}
\begin{arrayl}
H_0(\gamma)-2\pi\gamma &=&
\frac12 \,\ip^2 - (1-\cos\iphi) + \gamma\iphi, \\
H_0(\gamma)+h &=& \frac12 \,\ip^2 + \gamma\iphi,  \\
H_0(\gamma)+h &=& \frac12 \,\op^2 + \gamma\ophi,  \\
H_0(\gamma) &=&
\frac12 \,\op^2 - (1-\cos\ophi) + \gamma\ophi. 
\end{arrayl}
\end{equation}

Combining the equations in~\eqref{eq.qp}, we get expressions
for $\iphi$ and $\ophi$:
\begin{equation}\label{eq.phi}
\cos\iphi = 1-(h+2\pi\gamma)    \qmbox{and} 
\cos\ophi  = 1-h . 
\end{equation}
This is well-defined only if $0\leq h\leq
2(1-\pi\gamma)$. Hence there are maximal values for
$\gamma$ and $h$, given by
\[
\gamma_{\rm max} =\frac1\pi \qmbox{and} h_{\rm max} =
2(1-\pi\gamma).
\]  
If $\gamma>\gamma_{\rm max}$, then there is no blue curve that
intersects both the red unstable orbit and the magenta stable orbit,
hence no pinned fluxons exist if the applied bias current is larger
than $\gamma_{\rm max}$. If $h>h_{\rm max}$, then the blue curve does
not intersect the red manifold anymore.

Furthermore, $\iphi$ must lie on the red unstable manifold, hence
$\arcsin\gamma \leq \iphi\leq \phimax(\gamma)$, where
$\phimax(\gamma)$ is the maximal $\phi$-value of the orbit homoclinic
to $\arcsin\gamma$. As $h\in[0,2(1-\pi\gamma)]$, this implies that
there are two possible values for $\iphi$ and that $\ip>0$:
\[
\iphi = \pi \pm\arccos(2\pi\gamma-(1-h)) \qmbox{and}
\ip=\sqrt{2\left(H_0(\gamma)+h-\gamma\iphi\right)} .
\]
Note that the unstable manifold left of $\arcsin\gamma$ only
intersects with blue curves that have $\phi_x<0$, hence those orbits
can never connect to one of the stable manifolds of $2\pi+\arcsin\gamma$.

The point $(\ophi,\op)$ has to lie on the magenta stable manifolds, so
there can be up to five possible branches of solutions:
\begin{enumerate}
\item $\ophi = 2\pi - \arccos(1-h)$ {with} $\op >0$, {for all}
  $0\leq h\leq h_{\rm max}$;
\item $\ophi = 2\pi + \arccos(1-h)$ {with} $\op\geq0$, {for} $0\leq
  h\leq h_2$ {and} $\op<0$, {for} $h_2< h\leq h_{\rm max}$;
\item $\ophi = 2\pi + \arccos(1-h)$ {with~} $\op\geq0$, {for} $h_2<
  h\leq h_{\rm max}$;
\item $\ophi = 4\pi - \arccos(1-h)$ {with} $\op\geq0$, {for} $h_1<
  h\leq h_{\rm max}$;
\item $\ophi = 4\pi - \arccos(1-h)$ {with} $\op<0$, {for} $h_1< h\leq
  h_{\rm max}$.
\end{enumerate}
Here $h_2$ is the $h$-value such that the blue orbit intersects the
magenta manifolds at the equilibrium $(2\pi+\arcsin\gamma,0)$, i.e.,
$h_2(\gamma)=1-\sqrt{1-\gamma^2}$, and $h_1$ is such that the blue
orbit touches the magenta manifold at $(2\pi+\phi_{\rm
  max}(\gamma),0)$, the most-right point, thus $h_1(\gamma) =
1-\cos(\phi_{\rm max}(\gamma))$.  In all cases,
$|\op|=\sqrt{2\left(H_0(\gamma)+h-\gamma\ophi\right)}$.

To satisfy $h_2(\gamma)\leq h_{\rm max}(\gamma)$, we need that
$\gamma\leq\gamma_2= \frac{4\pi}{4\pi^2+1}\approx 0.3104$. If
$\gamma>\gamma_2$, then only pinned fluxons with $\ophi =
2\pi\pm\arcsin\gamma$ and $\op>0$ exist.  In order to have
$h_1(\gamma)\leq h_{\rm max}(\gamma)$, we need that
$\gamma\leq\gamma_1$, where $\gamma_1$ is the implicit solution of
$\cos\phimax(\gamma_1)+1=2\pi\gamma_1$, i.e., $\gamma_1\approx0.1811$.
If $\gamma>\gamma_1$, then no pinned fluxons with $\ophi =
4\pi-\arcsin\gamma$ exist.  On the intervals of common existence, we
have $0\leq h_2(\gamma)\leq h_1(\gamma)\leq h_{\rm max}(\gamma)$,
$h_1(\gamma_1)=h_{\rm max}(\gamma_1)$, $h_2(\gamma_2)=h_{\rm
  max}(\gamma_2)$, see Figure~\ref{fig.ham_d=0}.

\begin{figure}[htb]
  \centering
  \includegraphics[width=0.5\textwidth]{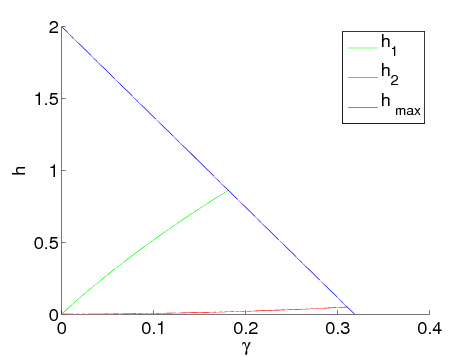}
  \caption{The extremal $h$-values $h_1(\gamma)$, $h_2(\gamma)$ and
  $h_{\rm max}(\gamma)$.} 
  \label{fig.ham_d=0}
\end{figure}

In Figure~\ref{fig.phase_d=0_g=.15_pinned1}, we have taken
$\gamma=0.15$ and $h=(h_1+h_{\rm max})/2$ and have plotted all five
possible pinned fluxons (i.e. all possibilities for $(\ophi,\op)$)
with $\iphi=\pi-\arccos(2\pi\gamma-(1-h))$. Obviously, five more
pinned fluxons with the same $(\ophi,\op)$ are possible with
$\iphi=\pi+\arccos(2\pi\gamma-(1-h))$.
\begin{figure}[htb]
\centering
\includegraphics[width=0.3\textwidth]{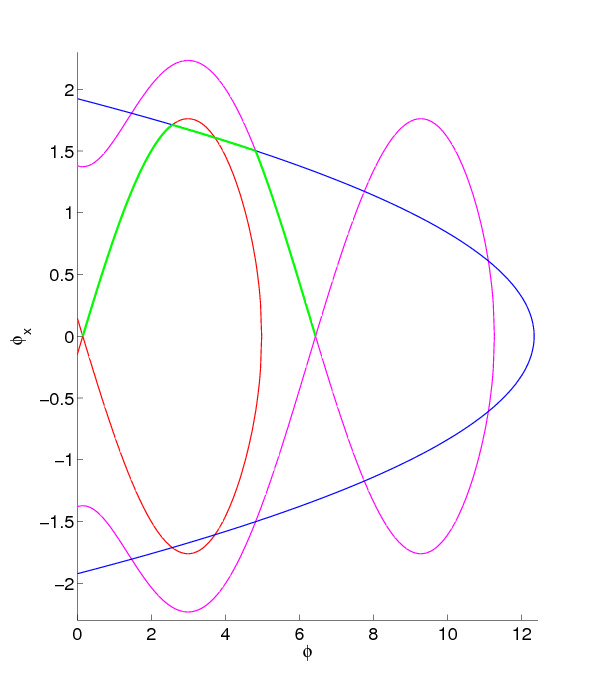}
\hfill
\includegraphics[width=0.3\textwidth]{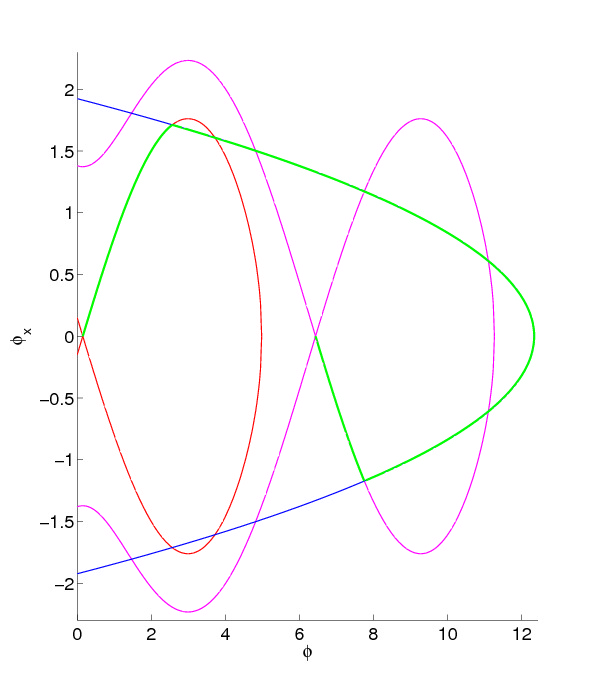}
\hfill
\includegraphics[width=0.3\textwidth]{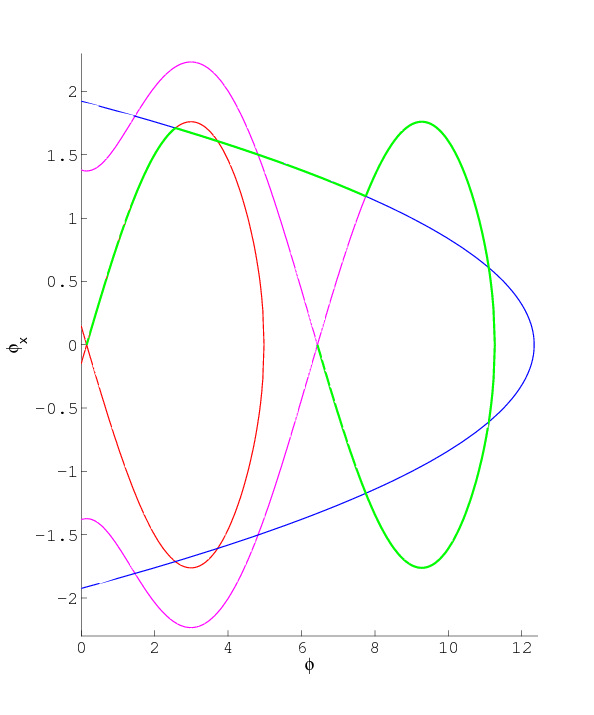}
\\
\includegraphics[width=0.3\textwidth]{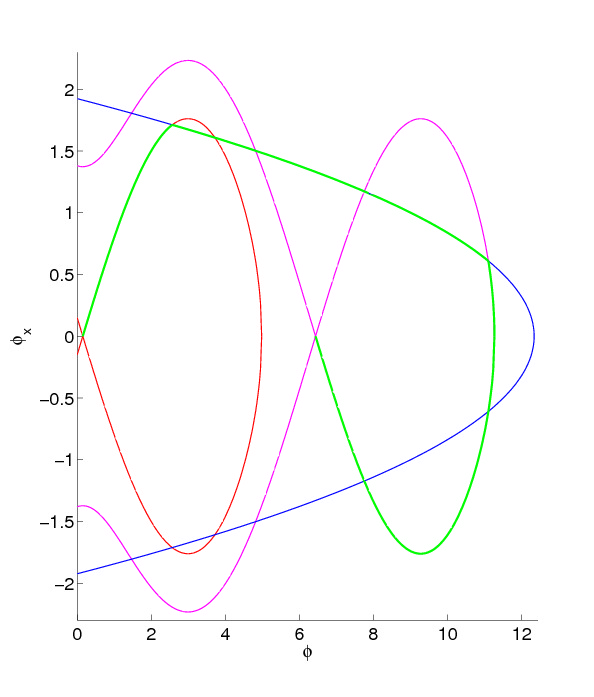}
\quad
\includegraphics[width=0.3\textwidth]{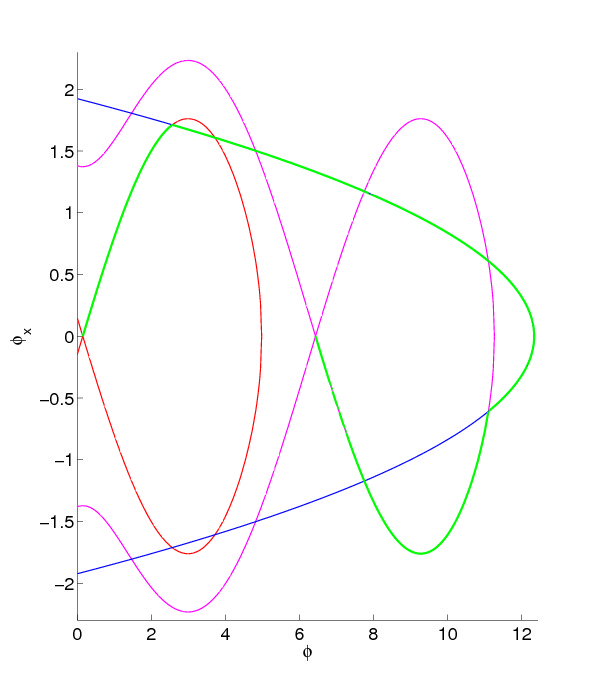}
\caption{The five pinned fluxons with
  $\iphi=\pi-\arccos(2\pi\gamma-(1-h))$ for $\gamma=0.15$, $d=0$ and
  $h=(h_1+h_{\rm max})/2$. Note that only the pinned fluxon in the
  first panel is monotonic. In the $L$-$h$ of
  Figure~\ref{fig.length_d_0}, the pinned fluxons in the first two
  panels are on the blue curve, the third one is on the red curve
  and the last two are on the green curve.} 
\label{fig.phase_d=0_g=.15_pinned1}
\end{figure}

To determine the length of the inhomogeneity for the pinned fluxons, we
use that on the orbits in the inhomogeneity (blue curves in the
phase portrait) $\phi$ and $\phi_x$ are related by
$|\phi_x|=\sqrt{2\left(H_0(\gamma)+h-\gamma\phi\right)}$. Integrating
this ODE, taking into account the sign of $\op$, we get that the
length of the pinned fluxons with $\op>0$ is given by
\begin{equation}\label{eq.length_d=0a}
2L = \frac{\sqrt2}{\gamma}\,\left[\sqrt{H_0+h-\gamma\iphi} -
  \sqrt{H_0+h-\gamma\ophi}\right] = \frac{\ip-\op}{\gamma}
\end{equation}
and for $\op<0$, we have
\begin{equation}\label{eq.length_d=0b}
2L = \frac{\sqrt2}{\gamma}\,\left[\sqrt{H_0+h-\gamma\iphi} +
  \sqrt{H_0+h-\gamma\ophi}\right] = \frac{\ip-\op}{\gamma}. 
\end{equation}
These lengths are plotted in Figure~\ref{fig.length_d_0} for
$\gamma=0.15$. 
\begin{figure}[htb] 
\centering
\parbox[b]{.6\textwidth}{\includegraphics[width=0.5\textwidth]{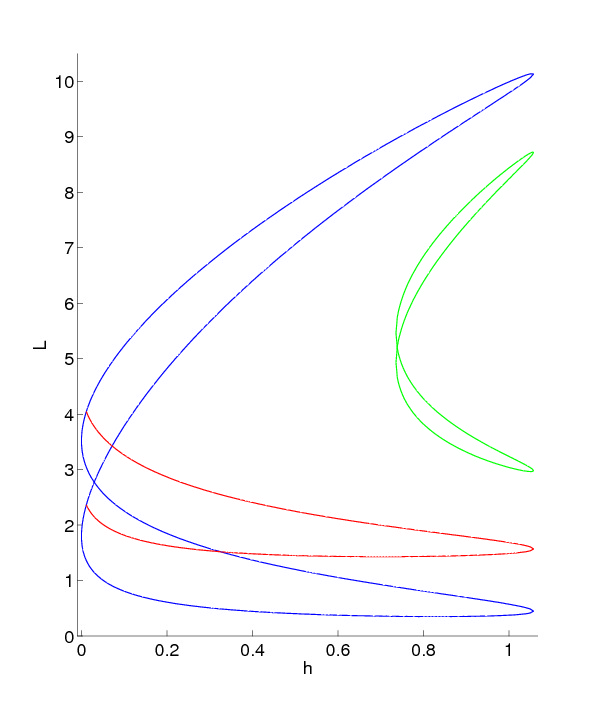}\vspace*{-6mm}}\hfill%
\parbox[b]{.4\textwidth}{\includegraphics[width=0.35\textwidth]{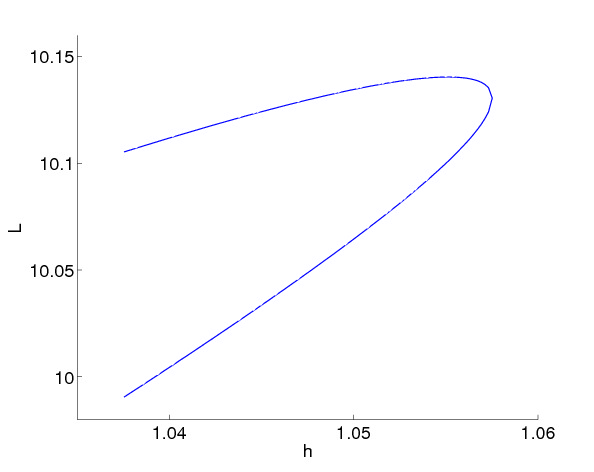}\\
\includegraphics[width=0.35\textwidth]{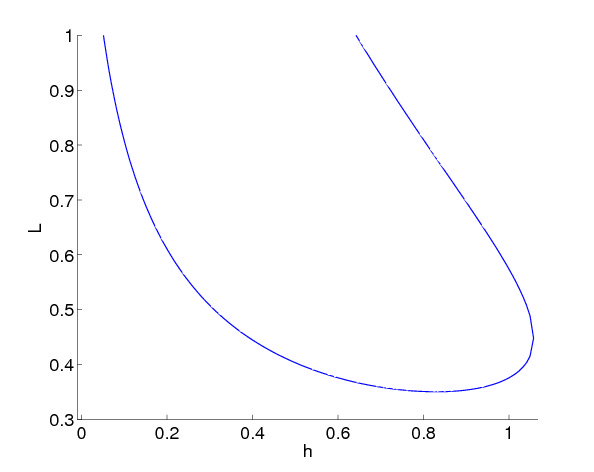}}
\caption{\label{fig.length_d_0} 
  The lengths of the pinned fluxons for $\gamma=0.15$ and $d=0$. The
  lengths of the pinned fluxons with $\ophi=4\pi-\arccos(1-h)$ are
  plotted in green (branches 4 and 5), the lengths of pinned fluxons
  with $\ophi=2\pi+\arccos(1-h)$ and $\op>0$ are in red (branch 3).
  The lengths of the remaining pinned fluxons (branches 1 and 2) are
  indicated by the blue curves. The panels on the 
  right zoom into the top and bottom and show that the minimal and
  maximal length are not obtained for $h_{\rm max}$, but a smaller
  value.}
\end{figure}
The blue curve is formed by the branches 1 and 2, the red curve  is
branch 3 and the green curve is formed by the branches 4 and 5.
This plot shows that there is a positive minimal and maximal length
for the inhomogeneity to sustain pinned fluxons. Inhomogeneities with
shorter or longer lengths will not be able to sustain pinned fluxons. %
Figure~\ref{fig.length_d_0} illustrates also that the maxima and
minima of the possible length of the inhomogeneity are attained inside
the interval $(0,h_{\rm max})$, not at the endpoints. These extremal
points will play an important role in the stability analysis as we
will see in the next section.

\begin{remark}\label{rem.branch}
  At $\ophi=2\pi+\arcsin\gamma$, i.e\ $h=h_2$, there is no
  bifurcation, although a solution appears/disappears. To see that
  this is not a bifurcation point, we look at the disappearing
  solution in the limit $h\downarrow h_2$. For $h\downarrow h_2$,
  there is one solution with $\ophi\approx2\pi+\arcsin\gamma$ and
  $\op<0$ and one solution with $\ophi\approx2\pi+\arcsin\gamma$ and
  $\op>0$. The solution with $\op<0$ remains very close to
  $2\pi+\arcsin\gamma$ for $x>L$. However, the solution with $\op>0$
  is tracking almost all of the homoclinic connection to
  $2\pi+\arcsin\gamma$. And in the limit $h\downarrow h_2$ this
  solution ``splits'' into the pinned fluxon with $\op=0$ and a full
  homoclinic connection (fluxon-antifluxon pair).
\end{remark}

In general, the derivation of the existence of the pinned fluxons
shows that for fixed $\gamma>0$ and $d=0$, there will always be a
strictly positive minimal and maximal length for the existence of
pinned fluxons.  From Figure~\ref{fig.ham_d=0}, it follows that the
green curve of pinned fluxons with $\ophi=4\pi-\arccos(1-h)$ is not
present if $\gamma>\gamma_1$. Similarly if $\gamma>\gamma_2$, the red
curve of pinned fluxons with $\ophi=2\pi+\arccos(1-h)$ and $\op>0$ are
not present.  Below we summarise the results for the existence
of the pinned fluxons with an induced current:
\begin{theorem}\label{th.existence_d=0}
  For $d=0$ and every $0<\gamma\leq\frac1\pi$, there are $L_{\rm
    min}(\gamma)$ and $L_{\rm max}(\gamma)$, such that for every
  $L\in(L_{\rm min},L_{\rm max})$, there are at least two pinned
  fluxons (at least one for $L=_{\rm min}$ or $L_{\rm max}$).
  Furthermore
\[
\lim_{\gamma\downarrow0}L_{\rm min}(\gamma)
  =0, \quad \lim_{\gamma\downarrow0}L_{\rm max}(\gamma)=\infty, 
\]
and
\[ \lim_{\gamma\uparrow1/\pi}L_{\rm min}(\gamma) =
  \lim_{\gamma\uparrow1/\pi}L_{\rm max}(\gamma) =\textstyle 
  \sqrt{\frac\pi2\left(\arcsin\frac1\pi+\sqrt{\pi^2-1}\right)} -
  \sqrt{\frac\pi2\left(\arcsin\frac1\pi+\sqrt{\pi^2-1}-\pi\right)}
  \approx 1.8.
\]
For given $L\in[L_{\rm min},L_{\rm max}]$, the maximum possible number
of simultaneously existing pinned fluxons is~6.  For
$\gamma>\frac1\pi$, there exist no pinned fluxons.
\end{theorem}

To relate the rich family of pinned fluxons which exists for
$\gamma>0$ with the unique pinned fluxons for $\gamma=0$, we have
sketched the $L$-$h$ curves for $\gamma=0.001$ in
Figure~\ref{fig.length_d_0_g_small}. 
\begin{figure}[htb]
  \centering
  \includegraphics[width=0.5\textwidth]{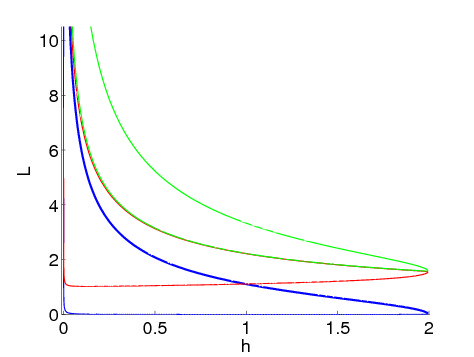}
  \caption{\label{fig.length_d_0_g_small}$L$-$h$ curves of the pinned
  fluxons for $\gamma=0.001$.}
\end{figure}
The bold blue curve converges clearly to the curve in
Figure~\ref{fig.length_d=0,g=0}. These are the lengths associated with
the pinned fluxons with $\iphi=\pi-\arccos(2\pi\gamma-(1-h))$ and
$\ophi=2\pi-\arccos(1-h)$.  Clearly there are some other convergent
$L$-$h$ curves as well. The length of the blue curve associated with
the pinned fluxons with $\iphi=\pi+\arccos(2\pi\gamma-(1-h))$ and
$\ophi=2\pi-\arccos(1-h)$ goes to zero as expected. The convergent red
and green curves can be associated with $4\pi$-fluxons in the case
$\gamma=0$. A $4\pi$-fluxon is a connection between $0$ and
$4\pi$. This is not possible without an inhomogeneity, but with an
inhomogeneity such connections are possible and some are stable. There
are four possible $4\pi$-fluxons if $\gamma=0$ and the green and red
curves converge to those solutions.  For more details,
see~\cite{mmath_chris}.

\subsection{Stability of the pinned fluxons with $d=0$}

As seen in the introduction, the stability of the pinned fluxons is
determined by the eigenvalues of the linearisation operator $\Lpin$ as
defined in~\eqref{eq.Lpin}. For $d=0$, the linearisation operator takes
the form
\[
  \Lpin(x;L,\gamma,0) =   \left\{
  \begin{arrayl}
 D_{xx} - \cos\phipin(x;L,\gamma,0),&& |x|>L;\\    
 D_{xx}, && |x|<L.
  \end{arrayl}\right.
\]
where $\phipin$ is one of the pinned fluxons found in the previous
section. 

When there is no induced current ($\gamma=0$), expressions for the
eigenvalues of $\Lpin$ can be found explicitly.  Recall that for $d=0$
and $\gamma=0$, there is a unique pinned fluxon for each length $L\geq
0$, see Lemma~\ref{lem.exist_d=0,g=0}.
\begin{lemma}\label{lem.stability_d=0,g=0}
  For $\gamma=0$ and $d=0$, the linear operator $\Lpin$ associated to
  the unique pinned fluxon in the defect with length~$L$ has a
  largest eigenvalue~$\Lambda_{\rm max}\in(-1,0)$ given implicitly by
  the largest solution of
\begin{equation}\label{eq.Lambda_max}
\textstyle
\begin{arrayl}
  \lefteqn{\textstyle-\mu\,\left[\mu+\frac12\,\sqrt{2(1+\cos\iphi)}\right]
  +\frac12\,(1-\cos\iphi) = } \\
&&\qquad\qquad 
-\sqrt{1-\mu^2} \,\left[\mu+\frac12\,\sqrt{2(1+\cos\iphi)}\right]\,
\tan \left(\sqrt{1-\mu^2}\frac{\pi-\iphi}{\sqrt{2(1-\cos\iphi)}}\right),
\end{arrayl}
\end{equation}
where $\mu=\sqrt{1+\Lambda_{\rm max}}\in(0,1)$ and the relation
between $\iphi$ and~$L$ is given in~\eqref{eq.in_g=0}
and~\eqref{eq.length_d=0}.
\end{lemma} 
In Figure~\ref{fig.d=0_stab}, $\Lambda_{\rm max}$ is sketched as
function of the half-length~$L$ of the pinned fluxon.  The proof of
Lemma~\ref{lem.stability_d=0,g=0} is quite technical; it is given in
appendix~\ref{app.Lambda1}.

\begin{remark}
  For $L$ large (hence $\iphi$ small), equation~\eqref{eq.Lambda_max}
  has more solutions. Hence for those pinned fluxons $\Lpin$ has some
  smaller eigenvalues in $(-1,0)$ too.
\end{remark}

\begin{corollary}\label{cor.stability_d=0,g=0}
If there is no induced bias current ($\gamma=0$) and the
microresistor has $d=0$, then the unique pinned fluxon in the defect
with length~$L$ is linearly and nonlinearly stable. The pinned
fluxon is asymptotically stable if $\alpha>0$.
\end{corollary}
 \begin{figure}[htb]
   \centering
 \includegraphics[width=0.6\textwidth]{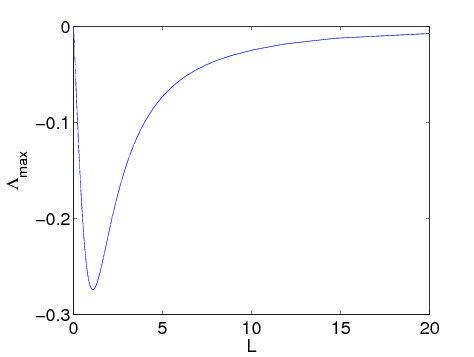}
   \caption{The largest eigenvalue of the linearised operator
     $\Lpin$ at $d=0$ and $\gamma=0$ as function of the half-length $L$
     of the inhomogeneity.}\label{fig.d=0_stab}
 \end{figure}
 
 Next we consider the case that there is an induced bias current
 ($\gamma>0$). In the previous section we have seen that in this case
 the pinned fluxons come in families, characterised by the blue, red
 and green curves in Figure~\ref{fig.length_d_0}. Locally along those
 curves, we can write either $L$ as a function of $h$, or, $h$ as a
 function of $L$. Along those curves, we will look for changes of
 stability, i.e., find whether the operator $\Lpin$ has an
 eigenvalue~0 (recall that eigenvalues of $\Lpin$ must be real. We
 will show that $\Lpin$ has an eigenvalue~0 if and only if along the
 $L$-$L$ curve we have $\frac{dL}{dh}=0$ or the pinned fluxon is
 isolated. Isolated pinned fluxons occur when $\gamma$ is maximal,
 i.e., $\gamma=\frac1\pi$ or when $\gamma=\gamma_1$, the maximal
 $\gamma$-value for which pinned fluxons with
 $\ophi=4\pi-\arccos(2\pi\gamma-1)$ exist.  This lemma is a special
 case of a more general theorem presented in~\cite{knight09}. The
 proof simplifies considerably in this case.
\begin{lemma}\label{th.pin_eigenval0}
  For any $\gamma\geq 0$, the linear operator
  $\Lpin(x;L,\gamma,0)$ has an eigenvalue zero if and only if
  \begin{itemize}
  \item $\frac{dL}{dh} = 0$;
  \item or $\gamma=\frac{1}{\pi}$ (this eigenvalue zero is
    the largest eigenvalue);
  \item or $\gamma=\gamma_1\approx0.18$, the solution of
    $\cos\phimax(\gamma_1)+1=2\pi\gamma_1$ (see
    section~\ref{sec.d=0,gamma>0}), and $\phipin$ is such that
    $\iphi=\pi$, $\ophi = 4\pi-\arccos(2\pi\gamma_1-1)=2\pi+\phi_{\rm
      max}(\gamma_1)$ (this eigenvalue zero is not the largest
    eigenvalue).
  \end{itemize}
\end{lemma}
\begin{proof}
  First we observe that differentiating~\eqref{eq.ham_ode_system} with
  respect to~$x$ shows that $\phipin'$ satisfies $\Lpin\phipin'=0$.
  However, it follows immediately from~\eqref{eq.ham_ode_system} that
  $\phipin'$ is not continuously differentiable, except when there
  exist $k_\pm\in\mathbb{N}$ such that $\iphi=k_-\pi$ and
  $\ophi=k_+\pi$. From the existence results, it follows that this
  happens only if $\gamma=\frac1\pi$. In this case, there is only one
  pinned fluxon and the blue curve in Figure~\ref{fig.length_d_0} has
  become a single point (there are no red or green curves).
  
  In all other cases, $\phipin'\not\in C^1(\mathbb{R}) \supset
  H^2(\mathbb{R})$ so $\phipin'$ is not an eigenfunction with the
  eigenvalue zero.  However, $\phipin'$ still plays a role in the
  eigenfunction related to any eigenvalue zero. Indeed, on both
  intervals $(\infty,-L)$ and $(L,\infty)$, the second order linear
  ODE $\Lpin \psi =0$ has two linearly independent solutions. As the
  asymptotic system is hyperbolic, one solution is exponentially
  decaying whilst the other is exponentially growing.  Thus if the
  linear operator~$\calL$ has an eigenvalue zero, then the
  eigenfunction in the intervals $(-\infty,L)$ and $(L,\infty)$ must
  be a multiple of the exponentially decaying solution. As $\phipin'$
  is exponentially decaying for $|x|\to\infty$ and satisfies ${\cal L}
  \phipin' = 0$ for $|x|>L$, it follows that for any eigenvalue zero,
  the eigenfunction must be a multiple of $\phipin'$ for $|x|>L$,
  unless $\phipin'\equiv 0$. The case $\phipin'\equiv 0$ happens only
  when $\ophi=2\pi+\arcsin\gamma$ and $x>L$. In this case, the
  appropriate eigenfunction is a multiple of
  $e^{-\sqrt[4]{1-\gamma^2}(x-L)}$.

  Next we look inside the inhomogeneity, i.e., $|x|<L$.  The
  linearised problem inside the defect for an eigenvalue zero can be
  solved explicitly and gives an eigenfunction of the form $A+B(x+L)$,
  with $A$ and $B$ free parameters and $|x|<L$.
  
  To conclude, if the linear operator $\Lpin$ has an eigenvalue zero,
  and $\ophi\neq 2\pi+\arcsin\gamma$ (we will consider this case
  later), then the eigenfunction is of the form
\[
\psi = \left\{
  \begin{arrayl}
    \phipin'(x) , &&  x<-L,\\
    A + B\, (x+L), && |x|<L,\\
    K\,\phipin'(x) , &&  x>L,
  \end{arrayl}\right.
\quad
\]
where $A$, $B$ and $K$ are free parameters. We have to choose the free
parameters such that $\psi$ is continuously differentiable at $\pm L$.
As there are only three free parameters and four matching conditions,
this will give us a selection criterion on the length~$L$ for which an
eigenvalue zero exists. 
The matching conditions are
\[
  A = \phipin'(-L^-), \quad 
  B = \phipin''(-L^-), \quad
  B =  K\phipin''(L^+), \qmbox{and} 
  A+2BL = K\phipin'(L^+), 
\]
where the notation $\phipin'(-L^-)=\lim_{x\uparrow -L}\phipin'(x)$,
$\phipin'(L^+)=\lim_{x\downarrow L}\phipin'(x)$, etc. Using that $p_{\rm
  in/out} = \phipin'(\mp L)$ and $\gamma+\phipin''(\pm L^\pm) = \sin\phi(\pm
L)=\sin\phi_{\rm in/out}$, this can be written as
\[
  A = \ip, \quad 
  B = \sin\iphi-\gamma, \quad
  B =  K(\sin\ophi-\gamma), \qmbox{and} 
  A+2BL = K\op. 
\]
Equations~\eqref{eq.length_d=0a} and~\eqref{eq.length_d=0b} show that
$L = \frac{\ip-\op}{2\gamma}$, hence the parameters are given by
\[
  A = \ip, \quad 
  B = \sin\iphi-\gamma, \qmbox{and} 
  K(\sin\ophi-\gamma) =  \sin\iphi-\gamma
\]
and the compatibility condition on $L$, or equivalently $h$, is
\begin{equation}\label{eq.comp}
0=\ip\sin\iphi(\sin\ophi-\gamma) - \op\sin\ophi(\sin\iphi-\gamma). 
\end{equation}
To derive this expression, we have multiplied the remaining equation
$[A+2BL=K\op]$ with $\gamma(\sin\ophi-\gamma)$.  This term would be
zero if $\sin\ophi=\gamma$, hence $\ophi=2\pi+\arcsin\gamma$ but this
case is not considered now.

For completeness, we also consider the case where we assume that the
eigenfunction vanishes for $x<-L$. If this is the case, then
matching at $x=-L$ gives immediately that $A=0=B$. Thus this leads to
a non-trivial eigenfunction only if $\phipin'(L)=0=\lim_{x\downarrow
  L}\phipin''(x)$. In other words, when $\phipin$ is a fixed point for
$x>L$. This happens only if $\ophi=2\pi+\arcsin\gamma$. This case
we will be considered later. 

Next we link the expression~\eqref{eq.comp} to the derivative of $L$
with respect to~$h$. As $L= \frac{\ip-\op}{2\gamma}$, the derivatives
of $\ip$ and $\op$ are needed. Differentiating~\eqref{eq.qp}
and~\eqref{eq.phi}, we get
\[
{\ip}\,\frac{d\ip}{dh} = {1-\gamma\iphi'(h)}, \quad
{\sin\iphi}\,\frac{d\iphi}{dh} = 1 \qmbox{and} 
{\op}\,\frac{d\op}{dh}
= {1-\gamma\ophi'(h)}, \quad {\sin\ophi}\,\frac{d\ophi}{dh} =  1.
\]
Thus differentiating $L= \frac{\ip-\op}{2\gamma}$ gives that
\begin{equation}\label{eq.dLdh}
\ip\sin\iphi\,\op\sin\ophi\,\frac{dL}{dh} = \frac{1}{2\gamma}\,
\left[\op\sin\ophi({\sin\iphi-\gamma}) -
  \ip\sin\iphi({\sin\ophi-\gamma})\right] 
\end{equation}
So we have shown that if $\ophi\neq2\pi+\arcsin\gamma$ and the
operator $\Lpin$ has an eigenvalue zero, then either
$\frac{dL}{dh}(h,\gamma)=0$ or $\ip\sin\iphi\,\op\sin\ophi=0$.
Considering $\ip\sin\iphi\,\op\sin\ophi=0$ in more detail, we get:
\begin{itemize}\addtolength{\itemsep}{-2mm} 
\item
$\sin\ophi=0$ would mean that $\ophi=2\pi$.
Going back to the compatibility condition~\eqref{eq.comp}, this
implies that $\gamma\ip\sin\iphi=0$, which only happens if also
$\sin\iphi=0$ or $\ip=0$. In the existence section we have seen
$\ip>0$, hence $\gamma\ip\sin\iphi=0$ can only happen if
$\iphi=\pi$, hence if $\gamma=\frac1\pi$;  
\item
  $\sin\iphi=0$ implies that $\iphi=\pi$.  Going back to the
  compatibility condition~\eqref{eq.comp}, this implies that
  $\gamma\op\sin\ophi=0$, which only happens if also $\sin\ophi=0$ or
  $\op=0$. Hence either $\gamma=\frac1\pi$ or $\gamma=\gamma_1$, as
  the case $\ophi=2\pi+\arcsin\gamma$ is excluded at this moment;
\item
$\ip\neq0$ as we have seen before;
\item
  $\op=0$ happens if $\ophi=2\pi+\arcsin\gamma$ or
  $\ophi=2\pi+\phi_{\rm max}(\gamma)$. Going back to the compatibility
  condition~\eqref{eq.comp}, this implies that
  $\ip\sin\iphi(\sin\ophi-\gamma)=0$. Since
  $\pi-\arcsin\gamma<\phi_{\rm max}(\gamma)<2\pi$, this implies this
  only happens if $\sin\iphi=0$, which case is considered before.
\end{itemize}
So altogether we have if $\ophi\neq2\pi+\arcsin\gamma$ and the operator
$\Lpin$ has an eigenvalue zero, then either
\begin{itemize}
\item 
$\frac{dL}{dh}(h,\gamma)=0$ or
\item 
$\iphi=\pi$ and $\ophi=2\pi$, which only happens when $\gamma=\frac
1\pi$. The eigenfunction in this case is $\phipin'$, which does not
have any zeros, hence the eigenvalue zero is the largest
eigenvalue. 
\item $\iphi=\pi$ and $\ophi=2\pi+\phi_{\rm max}(\gamma)$ (i.e.
  $\op=0$), which only happens if $\gamma=\gamma_1$. In this case the
  eigenfunction is $\phipin'$ for $x<L$ and
  $\frac{\gamma_1}{\gamma_1-\sin\phi_{\rm max}(\gamma_1)}\,\phipin'$
  for $x>L$. This eigenfunction has a zero at $x=L$, hence the
  eigenvalue zero is not the largest eigenvalue. Note that when
  $\gamma=\gamma_1$ the green $L(h)$ curve in
  Figure~\ref{fig.length_d_0} has degenerated to an isolated point
  related to the pinned fluxon~$\phipin$ considered in this case.
\end{itemize}

To show that the converse is true, we look at the three cases
$\frac{dL}{dh}(h,\gamma)=0$, $\gamma=\frac1\pi$ and $\gamma =
\gamma_1$ and $(\ophi,\op)=(2\pi+\phi_{\rm max},0)$. It is
straightforward to verify that the eigenfunctions as described earlier
can be constructed in those cases.

\smallskip Finally we look at the case $\ophi=2\pi+\arcsin\gamma$. In
this case $\gamma\leq \frac{4\pi}{1+4\pi^2}$ and
$h=h_2=1-\sqrt{1-\gamma^2}$.  Furthermore, the pinned fluxons
satisfies $\phipin'\equiv0$ for $x>L$.  In this case, the general form
of an eigenfunction for an eigenvalue zero is
\[
\psi = \left\{
  \begin{arrayl}
    \phipin'(x) , &&  x<-L,\\
    A + B\, (x+L), && |x|<L,\\
    K\,e^{-\sqrt[4]{1-\gamma^2}(x-L)} , &&  x>L,
  \end{arrayl}\right.
\quad
\]
where $A$, $B$ and $K$ are free parameters. We have to choose the free
parameters such that $\psi$ is continuously differentiable at $=\pm
L$, i.e.
\[
  A = \phipin'(-L^-), \quad 
  B = \phipin''(-L^-), \quad
  K= A+2BL, \qmbox{and} 
  B =  -K\sqrt[4]{1-\gamma^2}.
\]
As $L=\frac{\ip-\op}{2\gamma}=\frac{\ip}{2\gamma}$, this implies that
$A=\ip$, $B = \sin\iphi-\gamma$ {and} $K=\frac{\ip\sin\iphi}\gamma$,
with the matching condition 
\begin{equation}\label{eq.cond2}
\gamma(\sin\iphi-\gamma)
=-\sqrt[4]{1-\gamma^2} \sin\iphi\, \ip.
\end{equation}
If $\iphi=\pi+\arccos(2\pi\gamma-\sqrt{1-\gamma^2})$, then
$\sin\iphi<0$ and~\eqref{eq.cond2} cannot be satisfied as $\ip>0$ and
$\gamma>0$. If $\iphi=\pi-\arccos(2\pi\gamma-\sqrt{1-\gamma^2})$, then
the phase portrait in the existence section shows that
$\sin\iphi>\gamma$, thus  $(\sin\iphi-\gamma)>0$ and
again~\eqref{eq.cond2} cannot be satisfied. Thus no eigenvalue zero
can occur at $\ophi=2\pi+\arcsin\gamma$.  
\end{proof}

Lemma~\ref{th.pin_eigenval0} allows us to conclude the stability of
pinned fluxons.
\begin{theorem}\label{th.stability_d_0}
  For $d=0$, every $0<\gamma\leq\frac1\pi$, and every $L\in[L_{\rm
    min}(\gamma),L_{\rm max}(\gamma)]$, there
  is exactly one stable pinned fluxon. This pinned fluxon is linearly
  and nonlinearly stable (and asymptotically stable for $\alpha>0$).
  For $L$ sufficiently large $\left( L >
    \sqrt{\frac{\pi+\arcsin\gamma+\arccos(2\pi\gamma-\sqrt{1-\gamma^2})}%
      {2\gamma}}\right)$, the stable pinned fluxons are
  non-monotonic.
\end{theorem}
See Figure~\ref{fig.stab1} for an illustration of this theorem.
\begin{figure}[htb] \centering
\parbox[b]{.6\textwidth}{\includegraphics[width=0.45\textwidth]{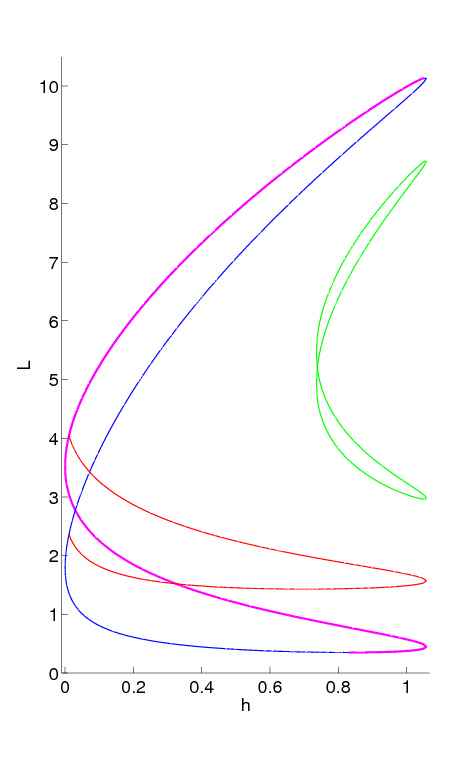}\vspace*{-6mm}}\hfill%
\parbox[b]{.4\textwidth}{\centerline{\includegraphics[width=0.4\textwidth]{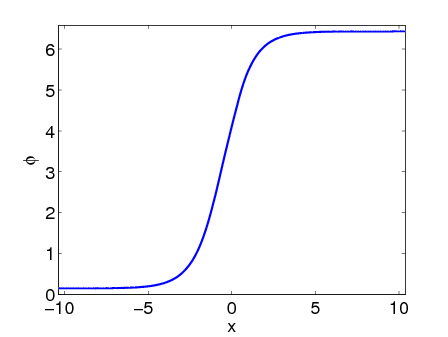}}\centerline{\includegraphics[width=0.45\textwidth]{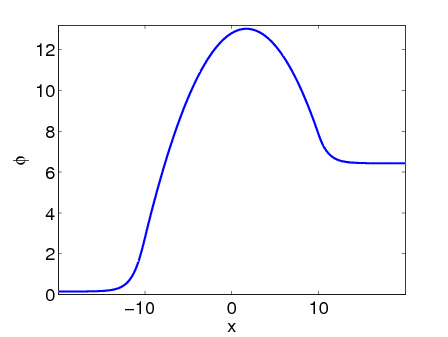}}}
  \caption{Stability for $d=0$ and $\gamma=0.15$. The bold magenta
    curve represents stable solutions, all other solutions are
    unstable. On the right there is an example of a stable monotonic
    pinned fluxon (at $L=0.38$) and a stable non-monotonic one (at
    $L=10$), Both stable pinned fluxons have $h=1$, i.e. they are near
    minimal respectively maximal length, which are at $L_{\rm
      min}=0.35$ and $L_{\rm max}=10.13$.}
  \label{fig.stab1}
\end{figure}

\begin{proof}
If $\gamma =\frac1\pi$, then only the inhomogeneity with half-length
exactly
$L=\sqrt{\frac\pi2\left(\arcsin\frac1\pi+\sqrt{\pi^2-1}\right)} -
\sqrt{\frac\pi2\left(\arcsin\frac1\pi+\sqrt{\pi^2-1}-\pi\right)}
\approx 1.8$ has a pinned fluxon. 
From Lemma~\ref{th.pin_eigenval0}, it follows that the
linearisation for this pinned fluxon has a largest eigenvalue 0, so this
pinned fluxon is linearly stable. 

In Corollary~\ref{cor.stability_d=0,g=0}, we have seen that the unique
pinned fluxons for $\gamma=0$ are stable. 

If $0<\gamma<\frac1\pi$, then there are at least two pinned fluxons if
$L\in(L_{\rm min},L_{\rm max})$, see Theorem~\ref{th.existence_d=0}.
As seen before, the $L$-$h$ curves for the pinned fluxons form three
isolated curves: $\ophi=4\pi-\arccos(1-h)$ (green curve), the (red)
curve of pinned fluxons with $\ophi=2\pi+\arccos(1-h)$ and $\op>0$
(exists for $h>h_2$), and the other pinned fluxons (blue curve). The
colour coding refers to Figures~\ref{fig.length_d_0}
and~\ref{fig.stab1}. The fluxons on the blue curve exist for all
$0\leq\gamma\leq\frac1\pi$; the existence of the other curves depends
on the value of $\gamma$.

The linearisation about the pinned fluxon at the minimum on the red
curve has an eigenvalue zero. At this point, the associated
eigenfunction is a multiple of $\phipin'$ for $x>L$.  On the red
curve, $\op>0$ and $\phipin'(x)<0$ for $x$ large. Thus this
eigenfunction has a zero. Using Sturm-Liouville theory, we can
conclude that the eigenvalue zero is not the largest eigenvalue. As
there is only one fluxon with $\frac{dL}{dh}=0$ on the red curves, all
pinned fluxons on the red curve are linearly unstable.

Similarly, the minimum and maximum on the green curve are associated
with pinned fluxons whose linearisation has an eigenvalue zero. Again,
the associated eigenfunction for $x>L$ is a multiple of $\phipin'$. As
for the red curve, at the minimum we have $\op>0$ and
$\phipin'(x)<0$ for $x$ large. Thus this eigenfunction has a zero and
we can conclude that the eigenvalue zero is not the largest
eigenvalue. The green curve is a closed curve with only two points
with $\frac{dL}{dh}=0$, so the eigenvalue zero at the maximum cannot
be the largest eigenvalue either. So we can conclude that all pinned
fluxons on the green curve are linearly unstable. 

Finally we consider the blue curve.  We use the stability of the
pinned fluxons at $d=0$, $\gamma=0$ to get a conclusion about the
stability of the pinned fluxons on this curve.  The solutions that can
be continued to $\gamma=0$ are the connections between
$\iphi=\pi-\arccos(2\pi\gamma-1+h)$ and $\ophi=2\pi-\arccos(1-h)$. Hence those
solutions are stable. Now using that zero eigenvalues can only occur
if $L(h)$ has a critical point, the blue curve can be divided in
stable and unstable solutions. The stable solutions are the part of
the curve $L(h)$ curve between the minimum and maximum that contains
the pinned fluxons with $\iphi=\pi-\arccos(2\pi\gamma-1+h)$ and
$\ophi=2\pi-\arccos(1-h)$. The pinned fluxons in the other part are
unstable as the zero eigenvalue is simple.  It can be verified that
the eigenfunctions related to the zero eigenvalues on this curve
do not have any zeroes indeed.

So altogether we can conclude that for each length there is exactly
one stable and at least one unstable solution. The stable fluxons are
non-monotonic if $L$ is larger than the length of the fluxon at
$h=h_2(\gamma)=1-\sqrt{1-\gamma^2}$ with
$\iphi=\pi-\arccos(2\pi\gamma-\sqrt{1-\gamma^2})$ and
$\ophi=2\pi+\arcsin\gamma$, hence
$L>\sqrt{\frac{\pi+\arcsin\gamma+\arccos(2\pi\gamma-\sqrt{1-\gamma^2})}%
  {2\gamma}}$.
\end{proof}


\section{General case ($d>0$)}

After analysing the existence and stability of pinned fluxons in
microresisors with the $d=0$ in full detail, in this section we will
sketch the existence and stability of the pinned fluxons for a general
microresistor or microresonator.

\subsection{Microresistors ($0<d<1$)}
The existence of pinned fluxons for $0<d<1$ follows from similar
arguments as for the case $d=0$. Using the matching of appropriate
solutions in the phase planes again, it can be shown that pinned
fluxons exist for $0\leq \gamma\leq\frac{1-d}{\pi}$.  The Hamiltonian
dynamics in the inhomogeneity satisfies the relation
\[
\frac 12 \phi_x^2 - d(1-\cos\phi) +\gamma\phi = H_0(\gamma) + h,
\]
where $h$ is a parameter for the value of the Hamiltonian as before.
The case $\gamma=0$ (no induced current) is more or less identical to
before, with a unique pinned fluxon for any $L>0$. For $\gamma>0$, a
similar calculation as in the case $d=0$ shows that there are two
possible entry angles:
\[\textstyle
\iphi = \pi-\arccos\left(\frac{2\pi\gamma-(1-d-h)}{1-d}\right) \qmbox{or}
\iphi = \pi + \arccos\left(\frac{2\pi\gamma-(1-d-h)}{1-d}\right)
\]
and up to three possible exit angles:
\[\textstyle
\ophi = 2\pi-\arccos\left(\frac{1-d-h}{1-d}\right),\quad
\ophi = 2\pi+\arccos\left(\frac{1-d-h}{1-d}\right),
 \qmbox{or}
\ophi = 4\pi-\arccos\left(\frac{1-d-h}{1-d}\right),
\]
with $0\leq h\leq 2(1-d-\pi\gamma)$.  If $\gamma>d>0$ (i.e., $d$ is
sufficiently close to zero), then there is still a minimal length
$L_{\rm min}(\gamma)>0$ and a maximal length $L_{\rm max}(\gamma)$ for
the inhomogeneity at which pinned fluxons can exist. However, if
$\gamma$ is less than $d$ ($0<\gamma\leq d$), then there is no upper
bound on the possible length of the inhomogeneity anymore, i.e.,
$L_{\rm max} =\infty$. This new phenomenon appears for
$\gamma/d\leq1$, due to the fact that now the dynamics in the
inhomogeneity have fixed points at
$(\phi,p)=(2k\pi+\arcsin(\gamma/d),0)$, $k\in \mathbb{Z}$. If $h$
corresponds to an orbit which contains such a fixed point, then the
length of an orbit with $\op<0$ goes to infinity. To illustrate this,
in Figure~\ref{fig.phase_d=0.2}, we have sketched the phase portraits
for $d=0.2$ and $\gamma=0.15<d$ and $\gamma=0.22>d$.
\begin{figure}[htb]
\centering
\includegraphics[width=0.49\textwidth]{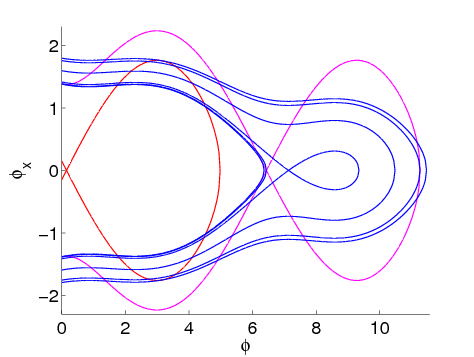}
\hfill
\includegraphics[width=0.49\textwidth]{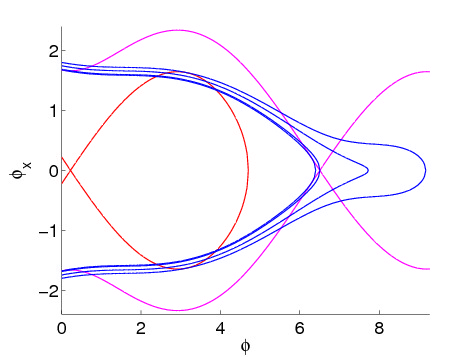}
  \caption{Phase portrait at $d=0.2$ and $\gamma=0.15$ (left) and
    $\gamma=0.22$ (right). Note that in the left graph, the third blue
    orbit has a fixed point. So the pinned fluxon with
    $\ophi=2\pi+\arccos\left(\frac{1-d-h}{1-d}\right)$ and $\op<0$
    does not exist for this $h$-value. Nearby pinned fluxons will be in
    a defect with a length that goes to infinity. In the right graph,
    there are no fixed points anymore as $\frac\gamma d>1$. Thus the
    defect lengths for which pinned fluxons exist are bounded.}
  \label{fig.phase_d=0.2}
\end{figure}

As before, the length of the inhomogeneity for the pinned fluxons
parametrised with~$h$ can be determined by using the relation
$|\phi_x| = \sqrt{2(H_0(\gamma) + h + d(1-\cos\phi) - \gamma\phi)}$
and integrating the ODE, taking care of the sign of $\phi_x$. The
resulting integrals cannot be expressed analytically in elementary
functions anymore, but they can be evaluated numerically. To
illustrate this, we have determined the $L$-$h$ curves as function
of~$h$ for $d=0.2$ and $\gamma=0.15$ ($\gamma<d$) and $\gamma=0.22$
($\gamma>d$).  The $L$-$h$ curves are presented in
Figure~\ref{fig.length_d=0.2}. Note the unbounded length curve for
$\gamma=0.15$.
\begin{figure}[htb]
\centering
\includegraphics[width=0.49\textwidth]{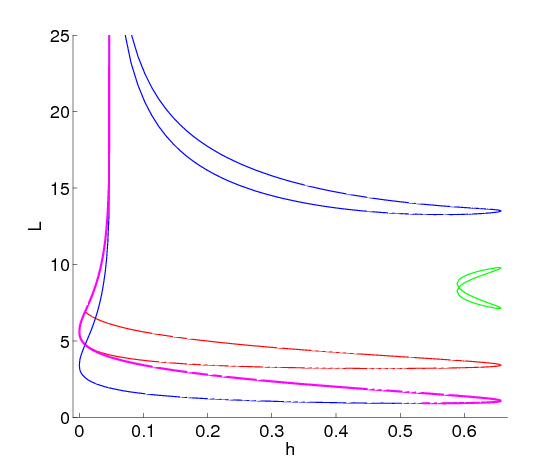}
\hfill
\includegraphics[width=0.49\textwidth]{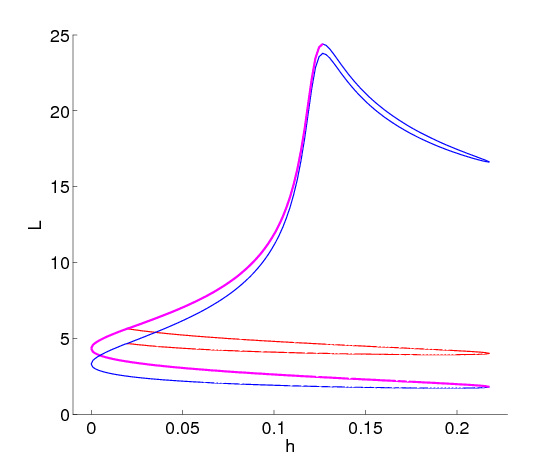}
  \caption{$L$-$h$ curves for $d=0.2$ and $\gamma=0.15$ (left) and
    $\gamma=0.22$ (right). For $\gamma=0.15$, the $L$-$h$ curves are
    unbounded as $\gamma$ is less than $d$. The color coding is as
    before, hence the bold magenta curve correspond to the stable fluxons.}
  \label{fig.length_d=0.2}
\end{figure}

In the following theorem, we summarise the existence of pinned fluxons
for $0<d<1$ and give their stability.
\begin{theorem}\label{th.exists_stab_0<d<1}
For $0<d<1$ and
\begin{itemize}
\item $\gamma=0$, there is a unique stable pinned fluxon for each $L\geq0$;
\item $0< \gamma\leq \min\left(d,\frac{1-d}\pi\right)$, there is a minimal length
  $L_{\rm min}(\gamma)>0$ such that for all $L> L_{\rm min}$ there
  exists at least two pinned fluxons (one for $L=L_{\rm min})$. For
  each $L\geq L_{\rm min}$, there is exactly one stable pinned fluxon;
\item $d<\gamma\leq \frac{1-d}{\pi}$, there are minimal and maximal
  lengths, $L_{\rm min}(\gamma)>0$ respectively $L_{\rm max}(\gamma)$ such
  that for all $L_{\rm min}< L< L_{\rm max}$ there exists at
  least two pinned fluxons, one pinned fluxon if $L$ is maximal or
  minimal, and no pinned fluxons exist for other
  lengths. For each $L_{\rm min}\leq L\leq L_{\rm max}$, there is
  exactly one stable pinned fluxon;
\item for $\gamma>\frac{1-d}{\pi}$, there exist no pinned fluxons.
\end{itemize}
\end{theorem}
Note that the third case will be relevant only if $0<d<\frac1{\pi+1}$.

\medskip 
To prove the stability result for the pinned fluxons, we will use
Theorem~\theoremknight from~\cite{knight09}.  In~\cite{knight09}, the
stability of fronts or solitary waves in a wave equation with an
inhomogeneous nonlinearity is considered. It links the existence of an
eigenvalue zero of the linearisation with critical points of the
$L$-$h$ curve. The proof has similarities with the proof of the case
$d=0$ in Lemma~\ref{th.pin_eigenval0}, but several extra issues have
to be overcome. Theorem~\theoremknight of~\cite{knight09}, applied to
our pinned fluxons for $0<d<1$, leads to the following lemma, which is
very similar to Lemma~\ref{th.pin_eigenval0} which holds for the
microresistor with $d=0$.
\begin{lemma}\label{th.pin_eigenval0_d<1}
  If $0<d<1$
, then the linear operator
  $\Lpin(x;L,\gamma,d)$ has an eigenvalue zero if and only if
  \begin{itemize}
  \item $\frac{dL}{dh} = 0$;
  \item or $\gamma=\frac{1-d}{\pi}$ (this eigenvalue zero is
    the largest eigenvalue);
  \item or $\gamma$ is such that it solves
    $(1-d)(\cos\phimax(\gamma)+1) = 2\pi\gamma$ and the pinned fluxon
    is such that $\iphi=2\pi+\phi_{\rm max}(\gamma)$ (this eigenvalue
    zero is not the largest eigenvalue).
  \end{itemize}
\end{lemma}
The verification of Lemma~\ref{th.pin_eigenval0_d<1} can be found
in~\cite[\textsection\exampleknight]{knight09}.
As far as the special cases in this lemma is concerned,
if $\gamma=\frac{1-d}\pi$ or $\gamma$ is such that it solves
$(1-d)(\cos(\phimax(\gamma)+1) = 2\pi\gamma$ and the pinned fluxon is
such that $\iphi=2\pi+\phi_{\rm max}(\gamma)$, then the pinned fluxon
under consideration corresponds an isolated ``green'' point and
$\frac{dL}{dh}$ does not exist. In the case of
$\gamma=\frac{1-d}{\pi}$, there is exactly one value of the length~$L$
for which there exists a pinned fluxon.  In the other case, there are
more pinned fluxons, but on other branches.
  In the case of an isolated pinned fluxon, either the derivative of the
  pinned fluxon is an eigenfunction with the eigenvalue zero or a
  combination of multiples of the derivative of the pinned fluxon is
  an eigenfunction.
%

The stability result of Theorem~\ref{th.exists_stab_0<d<1} follows
by combining Lemmas~\ref{lem.stability_d=0,g=0}
and~\ref{th.pin_eigenval0_d<1}. 
\begin{proofof}{Theorem~\ref{th.exists_stab_0<d<1}}
  The existence is described at the first part of this section, in
  this proof we focus on the stability. For $0\leq d<1$ and
  $\gamma=0$, there is a unique pinned fluxon for each length~$L$. It
  is straightforward to show that for each $0\leq d<1$, the length
  function~$L(h)$ is monotonic decreasing in~$h$.  Thus
  $\frac{dL}{dh}\neq 0$ and none of the pinned fluxons has an
  eigenvalue~zero. As all pinned fluxons are nonlinearly stable for
  $d=0$ (Lemma~\ref{lem.stability_d=0,g=0}) and no change of stability
  can happen, all pinned fluxons with $\gamma=0$ are nonlinearly
  stable for all $0\leq d<1$.
  
  If $0<d<1$ and $0<\gamma<\frac{1-d}{\pi}$, then the $L$-$h$ curve
  follows as a smooth deformation from the curve for $d=0$. And the
  unique stable pinned fluxon for each length follows.  
  
  If $0<d<1$ and $\gamma=\frac{1-d}{\pi}$, then the pinned fluxon is
  an isolated point and Lemma~\ref{th.pin_eigenval0_d<1} gives that it
  is stable.
\end{proofof}

\subsection{Microresonator ($d>1$)}

The existence results of pinned fluxons for $d>1$ are slightly
different from the ones for $d<1$. The main difference is the type of
solutions used in the inhomogeneous system. For $d<1$, we used
solutions that were part of unbounded orbits or homo/heteroclinic
orbits in the phase plane. For $d>1$, we have to use periodic orbits. 
The most simple way to understand this crucial difference between the
microresistor and the microresonator case is to consider the phase
portraits without applied bias current ($\gamma = 0$) -- see
Figure~\ref{fig.phase_gamma_0}.
\begin{figure}[htb]
\centering
\includegraphics[width=0.55\textwidth]{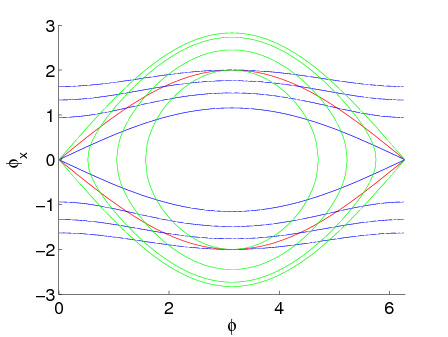} 
\caption{Phase portraits at $\gamma=0$ for various values of~$d$. The
 red curve is the heteroclinic connection at $d=1$. The blue curves
 are orbits for $d=\frac12$ and the green ones are orbits for $d=2$.}
  \label{fig.phase_gamma_0}
\end{figure}
When $d < 1$, respectively $d > 1$, the (red) heteroclinic orbit of
the system outside the inhomogeneity is outside, resp. inside, the
(blue resp.\ green) heteroclinic orbit of the system inside the
inhomogeneity -- see Figure~\ref{fig.phase_gamma_0}.  As a
consequence, a pinned defect can only be constructed with (unbounded)
orbits that are outside the (blue)inhomogeneous heteroclinic orbit in
the microresistor case, while one has to use bounded, periodic orbits
in microresonator case -- see the green lines in
Figure~\ref{fig.phase_gamma_0}.

One consequence is that if one solution for a inhomogeneity of a
certain length exists, then there are also solutions for
inhomogeneities with lengths that are this length plus a multiple of
the length of the periodic orbit. This implies that the number of
pinned fluxons for a defect of length $L$ may grow without bound as
$L$ increases -- which is very different from the microresistor
($d<1$). We will focus on the existence of solutions which use less
than a full periodic orbit as the other ones follow immediately from
this.

Using similar techniques as in the previous sections, it can be shown
that if $\widehat d$ is the solution of $-\frac{5\pi}{2}
+\arcsin\frac1{d} + \sqrt{d^2-1} + d - 1 =0$, ($\widehat d\approx
4.37$), then for $d> \widehat d$, pinned fluxons exist for any
$0\leq\gamma\leq1$. If $d\leq \widehat d$, then pinned fluxons exist
for $0\leq\gamma<\gamma_{\rm max}$, where $\gamma_{\rm max}(d)$ is the
(implicit) solution of $ -2\pi\gamma -\gamma\left(\arcsin\gamma
  -\arcsin\frac\gamma d\right) +\sqrt{d^2-\gamma^2}-\sqrt{1-\gamma^2}
+ (d-1)=0$.

 For illustration, phase portraits for $d=4$ and various
values of $\gamma$ are sketched in Figure~\ref{fig.phase_d=10}. This
illustrates that the solutions used in the inhomogeneous system (blue
lines) are all part of a periodic orbit. Note that for $\gamma>0$ both
unstable manifolds of $\arcsin\gamma$ and only the unbounded stable
manifold of $2\pi+\arcsin\gamma$ are used as opposed to the
microresistor case where only the bounded unstable manifold of
$\arcsin\gamma$ and both stable manifolds of $2\pi+\arcsin\gamma$ are
used.
\begin{figure}[htb]
\centering
\includegraphics[width=0.34\textwidth]{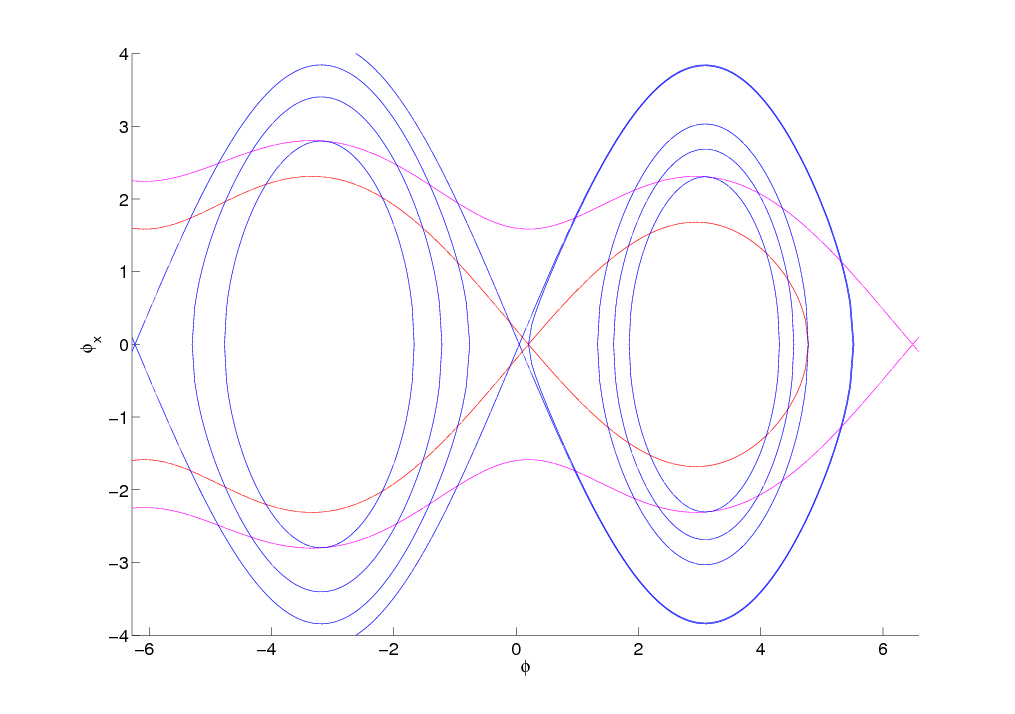}%
\hfill%
\includegraphics[width=0.34\textwidth]{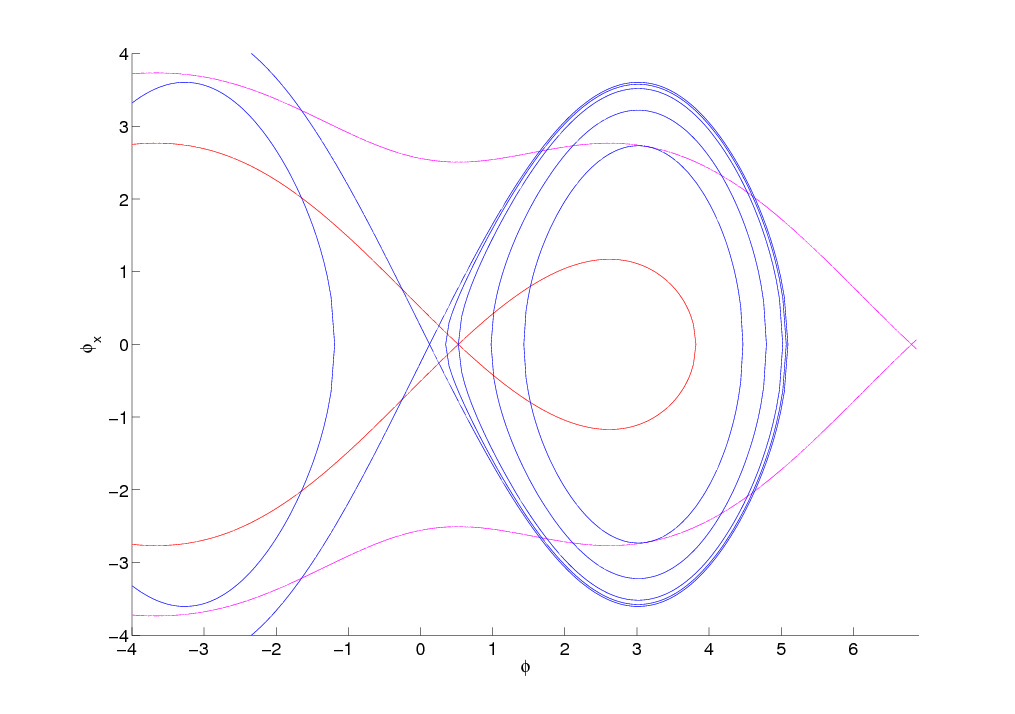}%
\hfill%
\includegraphics[width=0.3\textwidth]{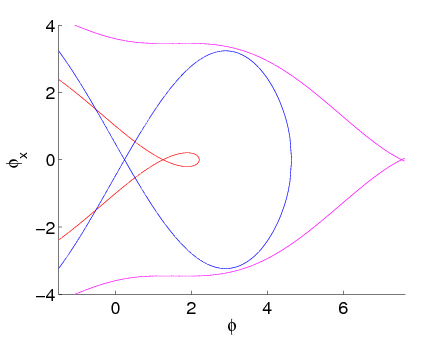}%
  \caption{Phase portrait at $d=4$ and $\gamma=0.2$ (left),
    $\gamma=0.5$ (middle) and $\gamma=0.95$ (right). As before, the red
    curves are the unstable manifolds to $\arcsin\gamma$ and the
    magenta ones are the stable manifolds to $2\pi+\arcsin\gamma$. The
    blue curves are orbits inside the inhomogeneity. The inner blue
    curve with angles between $0$ and $2\pi$ is the curve with the
    minimal $h$-value for which pinned fluxons exist. The blue curves
    can continue to relate to pinned fluxons up to (but not including)
    the blue homoclinic connection to $\arcsin(\frac{\gamma}d)$. Some
    of the periodic orbits with negative angles will also play a role
    in the construction of the pinned fluxons.  If
    $\gamma=0.95>\gamma_{\rm max}(d)\approx0.9$ (right plot), the blue
    homoclinic orbit (that encloses the red limit state) does not
    intersect the magenta stable manifold; illustrating that there
    cannot be pinned fluxons for $\gamma > \gamma_{\rm max}(d)$. }
  \label{fig.phase_d=10}
\end{figure}

As before, the dynamics in the inhomogeneity satisfies the relation
\[
\frac 12 \phi_x^2 - d(1-\cos\phi) +\gamma\phi = H_0(\gamma) + h,
\]
where $h$ is a parameter for the value of the Hamiltonian. Again it can
be shown that the entry and exit angles satisfy
\[
\cos\iphi = \frac{2\pi\gamma+d-1+h}{d-1} \qmbox{and}
\cos\ophi = \frac{d-1+h}{d-1},
\]
where now $-2(d-1)\leq h <h_{\rm max}$. Here $h_{\rm max}$ corresponds
to the $h$-value of the orbit homoclinic to $\arcsin\frac\gamma{d}$ in
the inhomogeneous system; it can be shown that $h_{\rm max}<0$. As we
use periodic orbits inside the inhomogeneity, the entry and exit
angles will differ by less than $2\pi$. For any $h$ value in
$[-2(d-1), h_{\rm max})$, there will be pinned fluxons with entry
angles between $\arcsin\frac\gamma d$ and $2\pi$.  For $\gamma$ small
relative to~$d$, entry angles less than $\arcsin\frac\gamma d$
are also possible and they can be related to smaller (more negative)~$h$
values.  The $p$-value for the exit points is always positive, while
the entry points can have both positive and negative $p$-values if the
entry angle is larger than $\arcsin\frac\gamma d$. The pinned fluxons
with entry angles less than $\arcsin\gamma$ have only negative
$\ip$-values and hence those pinned fluxons are non-monotonic and
``dip down''.

For $\gamma=0$, at least one pinned fluxon exists for each $L\geq0$.
If $L$ is sufficiently large, there will be more pinned fluxons.  This
is different to the case with $d<1$, where for $\gamma=0$, there is a
unique pinned fluxon for each length, it is due to the fact that the
pinned fluxons are buiklt from periodic orbits (that may be travelled
in various waus before leaving the inhomogeneity).  For $\gamma>0$,
there is minimum length $L_{\rm min}$ such that there are at least two
pinned fluxons for each length $L> L_{\rm min}$ (one for $L$ minimal).
The $L$-$h$ curves for $d=4$ and various $\gamma$ values are given in
Figure~\ref{fig.length_d=10}.
\begin{figure}[htb]
\centering
\includegraphics[width=0.32\textwidth]{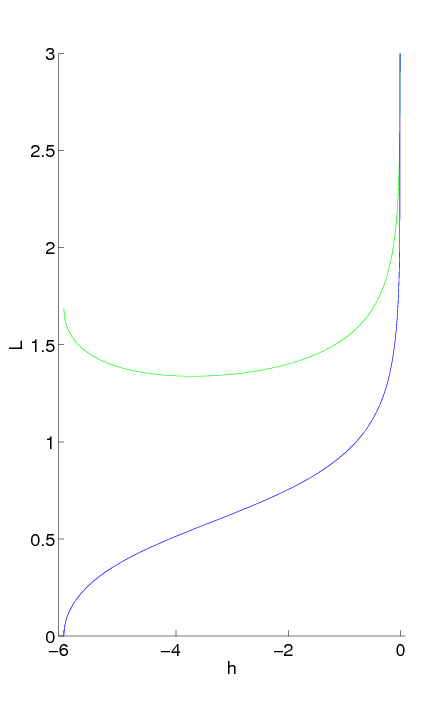}
\hfill
\includegraphics[width=0.32\textwidth]{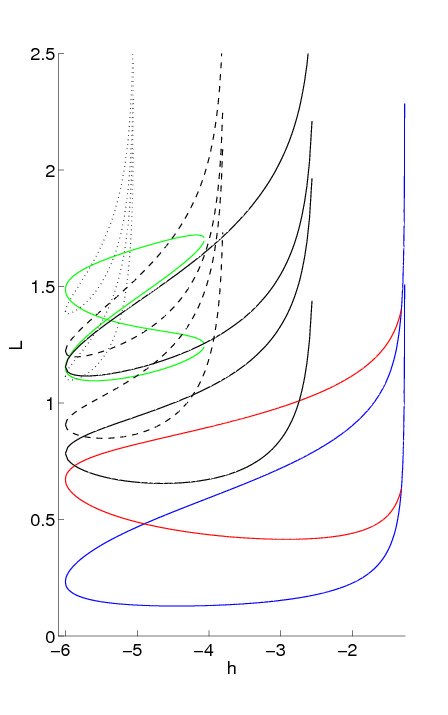}
\hfill
\includegraphics[width=0.32\textwidth]{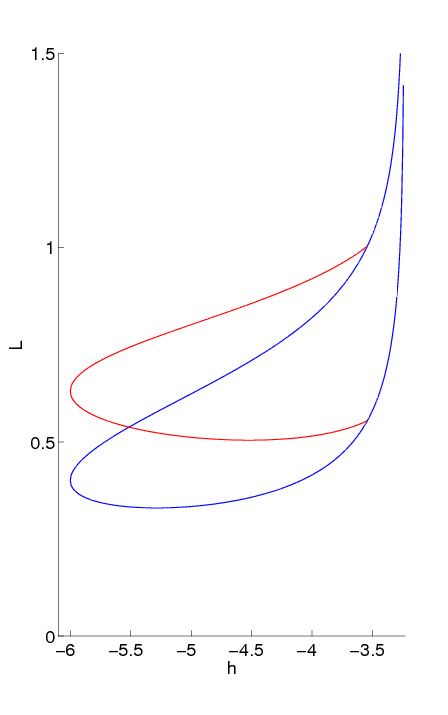}
\caption{$L$-$h$ curves at $d=4$ and $\gamma=0$ (left), $\gamma=0.2$
  (middle) and $\gamma=0.5$ (right). The blue and red curves are
  associated with pinned fluxons with $\arcsin\gamma/d<\iphi<\pi$. The
  red curve are pinned fluxons with $\ip<0$ and
  $\arcsin\gamma<\iphi<\pi$.  The pinned fluxons in the blue curve have
  $\ip>0$ for $\iphi>\arcsin\gamma$ and $\ip<0$ for $\iphi<\arcsin\gamma$.
  The green curves are associated with pinned fluxons with
  $\iphi>\pi$. In the middle panel ($\gamma=0.2$) there are also black
  curves, which are associated with pinned fluxons with
  $\iphi<\arcsin\frac\gamma d$. The solid black curves are lengths for
  pinned fluxons with $-2\pi+\arcsin\frac\gamma d<\iphi<\arcsin\frac\gamma
  d$, the dashed ones for pinned fluxons with
  $-4\pi+\arcsin\frac\gamma d<\iphi<-2\pi+\arcsin\frac\gamma d$, the
  dotted ones for pinned fluxons
  $-6\pi+\arcsin\frac\gamma d<\iphi<-4\pi+\arcsin\frac\gamma d$.}
  \label{fig.length_d=10}
\end{figure}
Only lengths of the pinned fluxons that use less than a full periodic
orbit are plotted.

In the following theorem, we summarise the existence of pinned fluxons
for $d>1$ and give their stability.
\begin{theorem}\label{th.exists_stab_d>1}
  Let $\widehat d$ be the solution of $-\frac{5\pi}{2}
  +\arcsin\frac1{d} + \sqrt{d^2-1} + d - 1 =0$ ($\widehat d\approx
  4.37$) and for $d>1$, let $\gamma_{\rm max}(d)$ be the (implicit)
  solution of $ -2\pi\gamma -\gamma\left(\arcsin\gamma
    -\arcsin\frac\gamma d\right)
  +\sqrt{d^2-\gamma^2}-\sqrt{1-\gamma^2} + (d-1)=0$.
\begin{itemize}
\item For $d>1$ and $\gamma=0$, there is at least one pinned fluxon for
  each $L\geq0$ and all pinned fluxons are unstable; 
\item For $1<d\leq\widehat d$ and $0<\gamma<\gamma_{\rm max}(d)$,
  there is a minimal length $L_{\rm min}(\gamma)>0$ such that for all
  $L> L_{\rm min}$ there exist at least two pinned fluxons (one for
  $L=L_{\rm min})$.  For each $L\geq L_{\rm min}$, there is at least
  one stable pinned fluxon. 
\item For $d>\widehat d$ and $0<\gamma\leq 1$,
  there is a minimal length $L_{\rm min}(\gamma)>0$ such that for all
  $L> L_{\rm min}$ there exist at least two pinned fluxons (one for
  $L=L_{\rm min})$.  For each $L\geq L_{\rm min}$, there is at least
  one stable pinned fluxon. 
\end{itemize}
\end{theorem}

In Figure~\ref{fig.length_d=10}, the stable pinned fluxons are the
pinned fluxons on the increasing part of the lower right blue curve.
Note that these pinned fluxons are non-monotonic past the meeting
point with the red curve, hence for most lengths.  The fluxons on the
other blue curve and red and green curves are unstable. As before, the
proof of the stability properties of Theorem~\ref{th.exists_stab_d>1}
is based on Theorem~\theoremknight from~\cite{knight09}.  The proof of
Theorem~\ref{th.exists_stab_d>1} is very similar to the proof of
Theorems~\ref{th.stability_d_0} and~\ref{th.exists_stab_0<d<1}. The
main difference is that we can not track our stability arguments back
to the case $d=0$ (i.e. Lemma~\ref{lem.stability_d=0,g=0}) as we did
before. The role of Lemma~\ref{lem.stability_d=0,g=0} will now be
taken over by Lemma~\ref{lem.approx_eigenval} in
Appendix~\ref{app.Lambda1}, in which it is explicitly established tht
the pinned fluxon on the blue curve has exactly one positive
eigenvalue for $\gamma=0$ and $d$ near one.

The stability of the fluxons on the black curves can not easily be
related to fluxons at $\gamma=0$ (they ``split'' in a homoclinic
``dip'' and a fluxon for $\gamma=0$). So a stability analysis for
this case goes outside the scope of this paper. In section~\ref{conc}, we
will show numerically that there are some stable fluxons on the black
curve.

\subsection{A microresonator approximating a localised
  inhomogeneity}\label{McLaughlin and Scott}

\newcommand{\dd}{\mathfrak{d}} 
\newcommand{\hh}{\mathfrak{h}} 
There have been quite a number of investigations on the influence of a
localised inhomogeneity, i.e., $D(x)=(1+\mu\delta(x))$ or
$D(x)=(1+\sum_{i=1}^N\mu_i\delta(x-x_i))$ in~\eqref{eq.junction}. In
this section we will confirm that our existence and stability results,
applied to short microresonators with large~$d$, reproduce in the limit for
$L\to 0$ and $d\to \infty$ the existence and stability results for
pinning by microshorts in~\cite{mcla78}.  In~\cite{mcla78} it is shown
that for $D(x)=(1+\mu\delta(x))$ and $\gamma$, $\mu$, and $\alpha$
of order $\eps$, with $\eps$ small and $\frac{\pi\gamma}{\mu}\leq
\frac{4}{3\sqrt{3}}+\mathcal{O}(\eps)$, there are one stable and
one unstable pinned fluxon, both approximated by
$\phi_0(x-X_0)+\mathcal{O}(\eps)$, where $X_0$ are the two solutions
of $-\frac{\pi\gamma}{2\mu}+\sech^2 X\tanh X$.

To approximate the localised inhomogeneities of $\delta$-function type with finite length ones,
we look at microresonators with length $L=1/(2\dd)$ and $d=1+\mu\dd$
for $\dd$ large. Thus the microresonators have short lengths and we can
restrict to pinned fluxons with
\[
\iphi = \arccos\left(\frac{2\pi\gamma+d-1+h}{d-1}\right) =
\arccos\left(1+\frac{2\pi\gamma+h}{\mu\dd}\right),\quad 
\ip>0,
\] 
{and}
\[
\ophi=\arccos\left(\frac{d-1+h}{d-1}\right) =
\arccos\left(1+\frac{h}{\mu\dd}\right). 
\]
Hence the pinned fluxons of~\cite{mcla78} correspond to solutions on the
lower blue curve in Figure~\ref{fig.length_d=10}.  Introducing $h=\mu
\dd \mathfrak{h}$, we get that $-2<\hh<0$ and we are interested in
$\hh$ away from $0$ as $\hh$-values close to zero correspond to long
lengths.  Using the expressions for $\iphi$, $\ip$ $\ophi$, and $\op$
and the ODE for the pinned fluxon, we can derive an asymptotic
expression for the length~$L(\hh)$ if $\dd$ is large and $\gamma$,
$\mu$ are order $\eps$, where $\eps$ is small:
\[
L (\hh)  = \frac{\pi\gamma}{-\hh\mu \dd \sqrt{2(2+\hh)}}
+\mathcal{O}(\dd^{-2} + \eps\dd^{-1}), \quad \eps,\,\dd^{-1} \to0.
\]
Thus $L(\hh)$ has a minimum at $\hh=-\frac{4}{3} +
\mathcal{O}(\dd^{-1} + \eps)$ and the condition
$L(\hh)=1/2\dd$ can be satisfied if the cubic
$\hh^2(2+\hh)=2\frac{\pi^2\gamma^2}{\mu^2} + \mathcal{O}(\dd^{-1} +
\eps)$ can be solved for some $\hh<0$. For $\hh<0$, this cubic has a
maximum at $\hh=-\frac{4}{3} + \mathcal{O}(\dd^{-1} + \eps)$, thus
$L(\hh)=\frac{1}{2\dd}$ has two solutions with $\hh$ between $-2$ and
$0$ iff $\frac{\pi\gamma}{\mu}\leq \frac{4}{3\sqrt{3}}
+\mathcal{O}(\dd^{-1} + \eps)$ (i.e., there are no solutions for
$\gamma/\mu$ too large). From the analysis in the previous section, we
can conclude that this corresponds to one stable pinned fluxon (least
negative value of $\hh$) and one unstable pinned fluxon.

Finally, for $\gamma=\mathcal{O}(\eps)$, with $\eps$ small, both the
unstable manifold to $\arcsin\gamma$ and the stable manifold to
$2\pi+\arcsin\gamma$ are close to the heteroclinic connection for the
unperturbed sine-Gordon equation. Thus for $x>L$, we have
$\phipin(x)=\phi_0(x-X_0)+\mathcal{O}(\eps)$, where $\phi_0$ is the
shape of the stationary fluxon in the sine-Gordon equation (and a
similar relation for $x<-L$). Substituting this into the equation for
$\ophi$, with $L=1/2\dd$ (hence $\hh$ is a solution of the cubic
introduced earlier), we get that $X_0$ is one of
the two solutions of $-\frac{\pi\gamma}{2\mu}+\sech^2 X\tanh X=0$.

\section{Conclusions and further work}
\label{conc}

This paper exhibits a full analysis for the existence and stability of
pinned fluxons in microresistors and microresonators for which the
Josephson tunneling critical current is modelled by a step-function.
It is shown that for fixed $d$ (Josephson tunneling critical current
inside the inhomogeneity) and fixed bias current $\gamma>0$, there is
an interval of lengths for which a rich family of pinned fluxons
exists. In the case when an induced current is present, there is a
lower bound on the length of inhomogeneities for which pinned fluxons
can exist.  If the inhomogeneity is too short, no pinned fluxons can
be sustained.  The lower bound on the length increases if the induced
current increases.  For microresistors with a sufficiently large
induced current, there is also an upper bound on the length for pinned
fluxons and the upper and lower bounds collide when the maximal value
of the induced current for which pinned fluxons can exist, is attained

Compared to the case of homogeneous wave equations, a new phenomenon
is observed: longer microresistors and microresonators have
non-monotonic stable pinned fluxons. In the case of microresistors
($d<1$), the non-monotonic stable pinned fluxons have a ``bump''
inside and behind the inhomogeneity and the values in the bump exceed
the asymptotic state $2\pi+\arcsin\gamma$. In the case of the
microresonators ($d>1$), the stable pinned fluxons have a ``dip''
before and near the inhomogeneity and the values in the dip are
between $\arcsin\gamma/d$ and $\arcsin\gamma$, i.e., below the left
asymptotic state.

To complement and illustrate the analytical results in the previous
sections, we have numerically solved the stationary
equation~\eqref{eq.ham_ode_system} for the pinned fluxons and the
corresponding linear eigenvalue problem~\eqref{eq.eigenval} using a
simple finite difference method and presented the results in
Figures~\ref{d0}--\ref{sum_d4}.
Without loss of generality as far as stability is concerned, we depict
the eigenvalues for $\alpha=0$, i.e., $\Lambda =\lambda^2$. Thus an
instability is indicated by the presence of a pair of eigenvalues
with non-zero real parts.

First, we consider the case of inhomogeneous Josephson junctions for a
microresistor with $d=0$. As is shown in Figure~\ref{fig.stab1}, when
$\gamma=0.15$ and the defect length parameter $L=4.2$, there are four
possible pinned fluxons. In Figure~\ref{d0}, the numerically obtained
profiles of pinned fluxons are a shown; all of them are clearly
non-monotonous. The insets show the eigenvalues of the fluxons in the
complex plane. Only one of them has no eigenvalues with non-zero real
parts, confirming that there is exactly one stable pinned fluxon,
which is non-monotonous for these parameter values. The
four pinned fluxons belong to two different families, the ones with
the smallest bump, i.e.\ panel (a), are on the blue curve and the
others, i.e.\ panel (b), on the green curve in Figure~\ref{fig.stab1}.
\begin{figure}[tbhp]
  \centering
\subfigure[]{\includegraphics[width=.49\textwidth]{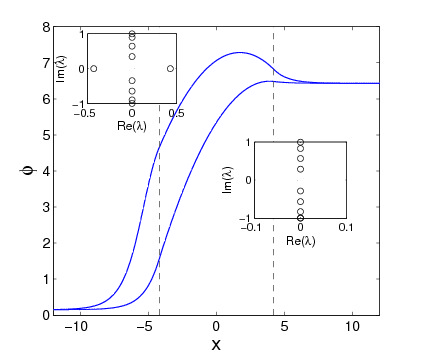}}%
\subfigure[]{\includegraphics[width=.49\textwidth]{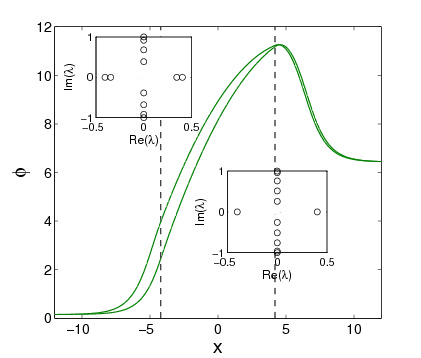}}
\caption{The four pinned fluxons admitted by the Josephson system with
  $d=0$ $L=4.2$, and $\gamma=0.15$. The insets show the eigenvalues of
  each fluxon; the top inset is related to the upper fluxon and the
  bottom inset to the lower fluxon. The vertical dashed lines show the
  edges of the defect.}
  \label{d0}
\end{figure}

In Figure~\ref{fig.stab1}, the existence and the stability of the
pinned fluxons for fixed $d$ and $\gamma$ are presented in the
$(h,L)$-plane and it is shown that each pair of the fluxons collide in
a saddle-node bifurcation at a critical $L$ for a fixed $\gamma$. To
complement these results, we take $L=4.2$ and numerically follow the
largest eigenvalue $\Lambda=\lambda^2$ of the various fluxons when the
induced current~$\gamma$ changes. The results are shown in
Figure~\ref{sum_d0}. As before, the colouring corresponds to the
colouring in Figure~\ref{fig.stab1}. Figure~\ref{sum_d0} shows that
there is a critical current for the existence of a pinned fluxon for a
given length and depth of the inhomogeneity. The blue and green
fluxons disappear in a saddle-node bifurcation. This happens at a
smaller value of $\gamma$ for the green fluxons (solutions in panel
(b) in Figure~\ref{d0}) than for the blue fluxons~(panel (a)). A
physical interpretation of the saddle-node bifurcation is that the
inhomogeneity is too short or long to pin a fluxon when the applied
current exceeds the critical value. For $\gamma=0.15$, there are no
red fluxons at this length, but they will exist for smaller values of
$\gamma$. The red fluxons disappear when the fluxon ``splits'' in a
homoclinic connection to $2\pi+\arcsin\gamma$ and a blue pinned
fluxon. Only one curve of red fluxons is visible. In theory, there is
a second curve, but this exist in a tiny $\gamma$-interval only and
hence is not visible.
\begin{figure}[bthp]
\centering
\includegraphics[width=.5\textwidth]{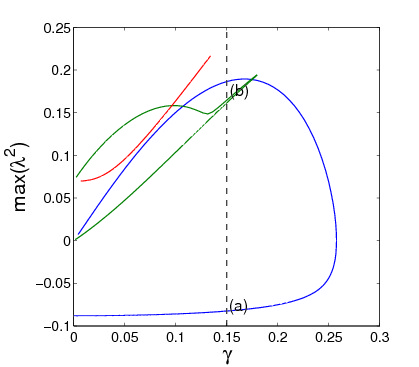}
\caption{The largest eigenvalue $\Lambda=\lambda^2$ of the various
  fluxons as function of the induced current~$\gamma$. The maximal
  eigenvalue at $\gamma=0.15$ of the fluxons in Figure~\ref{d0} is at
  the intersection between the curves and the vertical dashed
  line. Note that the blue and green fluxons disappear in a
  saddle-node bifurcation.}
\label{sum_d0}
\end{figure}

In Figures~\ref{d2}--\ref{sum_d4}, we consider the case of a
microresonator with $d=4$. 
From the middle panel in Figure~\ref{fig.length_d=10}, it follows that
there exist five pinned fluxons when $\gamma=0.2$, and $L=0.75$. In
Figure~\ref{d2} we show the numerically computed profiles of those
pinned fluxons and their eigenvalues, where the colouring is as in
Figure~\ref{fig.length_d=10}.  The blue non-monotonic fluxon is stable
while the blue monotonic one and red one are unstable. This confirms
our analytical findings (see Theorem~\ref{th.exists_stab_d>1}: there
is \emph{at least} one stable pinned fluxon). Moreover, it shows that
there can be more than one stable fluxon: one of the fluxons on the
black curve is stable too. So for $d>1$, there is bi-stability for
some values of $L$ and~$\gamma$.
\begin{figure}[bth]
  \centering
\subfigure[]{\includegraphics[width=0.32\textwidth]{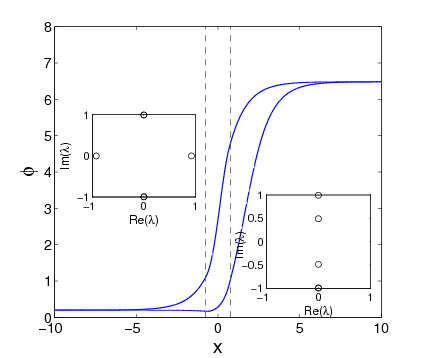}}
\subfigure[]{\includegraphics[width=0.32\textwidth]{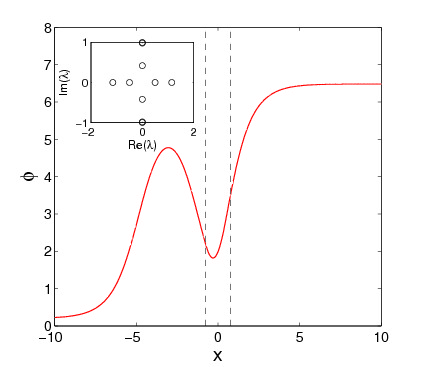}}
\subfigure[]{\includegraphics[width=0.32\textwidth]{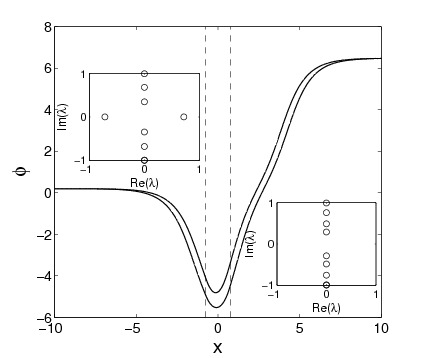}}
\caption{The same as in Figure~\ref{d0}, but for a microresonator with
  $d=4$, $L=0.75$, and $\gamma=0.2$, where there are five pinned
  fluxons.}
  \label{d2}
\end{figure}

In Figure~\ref{sum_d4} we also present the critical eigenvalues of the
five fluxons as a function of $\gamma$ when $L=0.75$ is
fixed. Similarly as in Figure \ref{sum_d0}, the pairs of blue and
black fluxons collide in a saddle-node bifurcation, while the red
fluxon breaks up at the maximal value of~$\gamma$.

\begin{figure}[bthp]
\centering
{\includegraphics[width=0.5\textwidth]{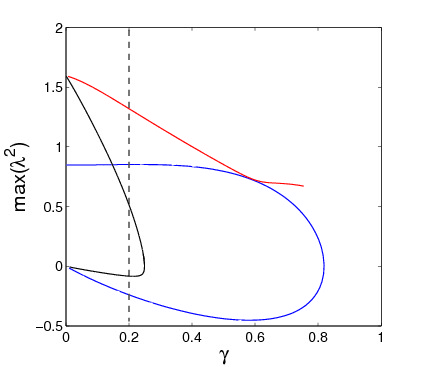}}
\caption{The same as in Figure \ref{sum_d0}, but for the five pinned fluxons in Figure \ref{d2}.}
\label{sum_d4}
\end{figure}


\medskip
For future research, it is of interest to expand our study to the case
of two-dimensional Josephson junction with inhomogeneities. A
particular example is the so-called window Josephson junction, which
is a rectangular junction surrounded by an inhomogeneous 'idle' region
with $d=0$. The interested reader is referred to
\cite{bena02,bena02b,capu00} and references therein for reviews of
theoretical and experimental results on window Josephson
junctions. Recently, fluxon scatterings in a 2D setup in the presence
of a non-zero defect has been considered as well in~\cite{piet05}. 

One can also apply our method to study the existence of trapped solitons by inhomogeneities in Schr\"odinger
equations, such as pinned optical solitons in a nonlinear Bragg media
with a finite-size inhomogeneity (see, e.g., \cite{good07b} and
references therein) and trapped Bose-Einstein condensates by a finite
square-well potential (see, e.g., \cite{carr01,park04}). In general,
the ideas presented in this paper are applicable to any system with
locally (piecewise constant) varying parameters in the equations as
can be seen in preprints by some of us~\cite{MJS} and~\cite{HDKNU}.

Finally, the simulations in section~\ref{sec.sim} show how
inhomogeneities can capture travelling fluxons. This suggests that the
pinned fluxons analysed in this paper can be attractive or repelling,
just as observed in~\cite{mcla78} in case of the localised
inhomogeneities.  We are currently investigating the attractive and
repelling interaction of the travelling fluxons with the pinned
fluxons and will report on this in a future paper.

\bigskip\noindent
\textbf{Acknowledgement} {We would like to thank Daniele Avitabile for
  the use of his suite of simulation codes.}

\appendix
\section{Largest eigenvalue of linearisation with no induced
  current}\label{app.Lambda1} 
\begin{proofof}{Lemma~\ref{lem.stability_d=0,g=0}}
  Let $\gamma=0$, $d=0$, and fix the length~$L$ of the inhomogeneity.
  We denote the unique pinned fluxon with length~$L$ by $\phipin(x)$
  (suppressing all other parameters). From~\eqref{pf}, we see that
  $\phipin$ equals the sine-Gordon fluxon outside the inhomogeneity
  ($|x|>L$) and the linearisation about the sine-Gordon fluxon is
  well-studied. The shifted pinned fluxon~$\phipin(x)-\pi$ is an odd
  function, hence a quick inspection shows that the operator
  $\Lpin(x)$ is even in~$x$ (we suppress all other parameters in
  $\Lpin$). All eigenvalues of $\Lpin$ are simple, thus $\Lpin(x)$
  being even implies that all eigenfunctions are odd or even. The
  eigenfunction for the largest eigenvalue does not have any zeroes,
  thus this eigenfunction is even.

For fixed $\Lambda$, the linear ODE associated with
  $(\Lpin-\Lambda)\Psi =0$ has two linearly independent solutions.
  The asymptotic limits of $\phipin$ for $x\to\pm\infty$
  correspond to saddle points in the ODE~\eqref{eq.ham_ode_system} and
  the decay rate to these fixed points is like $e^{-x}$. This implies
  that for $\Lambda>-1$, there is one
  solution of the ODE $(\Lpin-\Lambda)\Psi =0$ that is exponentially
  decaying at $+\infty$ and there is one solution that is
  exponentially decaying at $-\infty$. We denote
  the exponentially decaying function at $-\infty$ by
  $v_-(x;L,\Lambda)$.
  
  In~\cite{mann97}, the linearisation of the sine-Gordon equation
  about the the fluxon~$\phi_0$ is studied in great detail. Using the
  results in this paper, we can derive an explicit expression for the
  solutions~$v_-(x;L,\Lambda)$ (see also~\cite{ddvgs07}), they are
\[
\begin{arrayl}
v_-(x;L,0) &=& \sech(x+x^*(L)),\quad x<-L\\
v_-(x;L,\Lambda) &=& e^{\mu (x+x^*(L))} \, [\tanh(x+x^*(L))-\mu] , \quad x<-L 
\qmbox{where} \mu = \sqrt{\Lambda+1},
\end{arrayl}
\]
where $x^*(L)$ is given in Lemma~\ref{lem.exist_d=0,g=0}.
In the inhomogeneity, the linearised operator is simply
$\Lpin=D_{xx}$, hence the even solutions of $\Lpin-\Lambda$ are
\[
\begin{arrayl}
  v_{\rm inhom}(x;\Lambda) &=& A\cos(\sqrt{-\Lambda} x), 
  \quad |x|<L, \qmbox{if $\Lambda<0$;}\\
  v_{\rm inhom}(x;0) &=& A, \quad |x|<L; \\
  v_{\rm inhom}(x;\Lambda) &=& A\cosh(\sqrt{\Lambda} x), 
  \quad |x|<L, \qmbox{if $\Lambda>0$.}
\end{arrayl}
\]
To have a continuously differentiable solution of $(\Lpin-\Lambda)
\psi=0$ in $H^2(\mathbb{R})$, we have to match $v_-$ and $v_{\rm
  inhom}$ and its derivatives at $x=-L$ (the conditions for $x=L$
following immediately from this as the eigenfunction is even). This
gives:
\begin{itemize}
\item If $\Lambda=0$ (thus $\mu=1$):
\[
A =  \sech\,\xi^* \qmbox{and}
0 = - \sech\,\xi^*\tanh\,\xi^*
\]
with $\xi^*=-L+x^*(L)$. This implies that $\xi^*=0$ and $A=0$. From
the relation for $x^*(L)$ in Lemma~\ref{lem.exist_d=0,g=0}, it follows
$\xi^*\neq 0$ only if $L=0$, hence when there is no inhomogeneity.
This confirms that the stationary sine-Gordon fluxon (the pinned
fluxon for $L=0$) has an eigenvalue zero, but none of the pinned
fluxons with $L>0$ will have an eigenvalue zero for its linearisation
$\Lpin$.
\item If $\Lambda>0$ (thus $\mu>1$), with $y^*=L\sqrt{\mu^2-1}$ and
  again $\xi^*=-L+x^*(L)$:
\[
\begin{arrayl}
  A\cosh y^*    &=&e^{\mu
    \xi^*} \, [\tanh\xi^*-\mu] \\
  - \sqrt{1-\mu^2}\,A\sinh y^* &=&
  e^{\mu \xi^*} \, [\mu(\tanh\xi^*-\mu) +\sech^2\xi^*]
\end{arrayl}
\]
Hence  $\mu$ (thus $\Lambda$) is determined by
\[
\mu\,[\tanh\xi^*-\mu] +\sech^2\xi^* = 
-\sqrt{\mu^2-1} \,[\tanh\xi^*-\mu]\,\tanh y^*.
\]
Using Lemma~\ref{lem.exist_d=0,g=0}, this can be written as a
relation between $\mu$ and $\iphi$ (and hence $\mu$ and $L$ as there
is a bijection between $\iphi\in(0,\pi)$ and $L> 0$):
\begin{equation}\label{eq.d=0_stab>}
\begin{arrayl}
  \lefteqn{\textstyle-\mu\,\left[\mu+\frac12\,\sqrt{2(1+\cos\iphi)}\right]
  +\frac12\,(1-\cos\iphi) = } \\
&&\qquad\qquad 
\sqrt{\mu^2-1} \,\left[\mu+\frac12\,\sqrt{2(1+\cos\iphi)}\right]\,
\tanh \left(\sqrt{\mu^2-1}\frac{\pi-\iphi}{\sqrt{2(1-\cos\iphi)}}\right).
\end{arrayl}
\end{equation}
It can be seen immediately that the right-hand side
of~\eqref{eq.d=0_stab>} is positive. The left-hand side
of~\eqref{eq.d=0_stab>} is always negative for $\mu>1$ as
\[
-\mu\,(\mu+T) + 1-T^2 = 1-\mu^2 -\mu T-T^2 \leq -T-T^2<0 ,
\]
where we wrote $T=\frac12\,\sqrt{2(1+\cos\iphi)}$, hence
$\frac12\,(1-\cos\iphi)=1-T^2$.  Thus~\eqref{eq.d=0_stab>} has no
solutions and there do no exist any eigenvalues $\Lambda>0$.
\item If $-1<\Lambda<0$ (thus $0<\mu<1$), again with $\xi^*=-L+x^*(L)$
  and  now  $y^*=L\sqrt{1-\mu^2}$:
\[
  A\cos y^*   = e^{\mu
    \xi^*} \, [\tanh\xi^*-\mu] \qmbox{and}
  \sqrt{1-\mu^2}\,A\sin y^* = 
  e^{\mu \xi^*} \, [\mu(\tanh\xi^*-\mu) + \sech^2\xi^*].
\]
Hence $\mu$ (thus also $\Lambda$) is determined by
\[
\mu\,[\tanh\xi^*-\mu] +\sech^2\xi^* = 
\sqrt{1-\mu^2} \,[\tanh\xi^*-\mu]\,\tan y^*.
\]
Using the same relations as before, this can be  written as a relation
between $\mu$ and $\iphi$:
\[\textstyle
\begin{arrayl}
  \lefteqn{\textstyle-\mu\,\left[\mu+\frac12\,\sqrt{2(1+\cos\iphi)}\right]
  +\frac12\,(1-\cos\iphi) = } \\
&&\qquad\qquad 
-\sqrt{1-\mu^2} \,\left[\mu+\frac12\,\sqrt{2(1+\cos\iphi)}\right]\,
\tan \left(\sqrt{1-\mu^2}\frac{\pi-\iphi}{\sqrt{2(1-\cos\iphi)}}\right).
\end{arrayl}
\]
Bringing all terms to the left and writing
$T(L)=\frac12\,\sqrt{2(1+\cos\iphi(L))}\in(0,1)$ gives on the left
\[
F(L,\mu) :=  -\mu\,\left[\mu+T\right] +1-T^2 + 
\sqrt{1-\mu^2} \,\left[\mu+T\right]\,
\tan \left(\sqrt{1-\mu^2}L\right).
\]
Taking $\mu=1$ in this expression gives $F(L,1) = -T-T^2\leq 0$.  If
$L<\frac\pi2$, then $F(L,0) = 1-T^2 + T \tan L>0$ as $T\in (0,1)$.  If
$L\geq \frac\pi2$, then $\iphi<\frac\pi2$ and $T>\frac12\sqrt2$, thus
$F(L,\frac{\sqrt{L^2-(\pi/2-\eps)^2}}{L}) \geq -2 + 
\frac{(\pi/2-\eps)\sqrt 2}{2L}\,\tan(\frac\pi2-\eps) =
\mathcal{O}(\frac{1}{L\eps})$, for $\eps\to0$. As $L$ is fixed, we can
choose $\eps$ such that this expression is positive. Thus we can
conclude that for all $L>0$, there is at least one $\mu\in(0,1)$
that solves $F(L,\mu)=0$. If $L$ gets very large, then there will be many
solutions, but we are interested in the largest one.
\end{itemize}
\end{proofof}

\begin{lemma}\label{lem.approx_eigenval}
  For $\gamma=0$ and $d=1+\eps$ with $\eps$ small, the linearisation
  $\Lpin(x;L,0,1+\eps)$ about the monotone pinned fluxon
  $\phipin(x;L,0,1+\eps)$ has a largest eigenvalue of the form
  $\eps\Lambda_1+\mathcal{O}(\eps^2)$ with
\[
\Lambda_1 = 
\frac{\sech^2L\,\left[-L^2\sech^4L(1+\tanh^2L) +
2 L\tanh L(\sech^4L +2(1+\sech^2L))+ \tanh^2L(6+\sech^2L)\right]}
{16(L\sech^2L+\tanh L)}.
\]
See Figure~\ref{fig.eigenval_d=1} for a sketch of~$\Lambda_1$.
Furthermore, if there are any other eigenvalues, then they must be
near $-1$.  Thus for $\gamma=0$ and $d$ close to~1, the monotone
pinned fluxons with $d>1$ are linearly unstable. The nonlinear
stability of Theorem~\ref{th.exists_stab_0<d<1} is confirmed by the
sign of $\Lambda_1$ for $d<1$.
\end{lemma}
\begin{figure}[htb]
  \centering
  \includegraphics[height=0.5\textwidth]{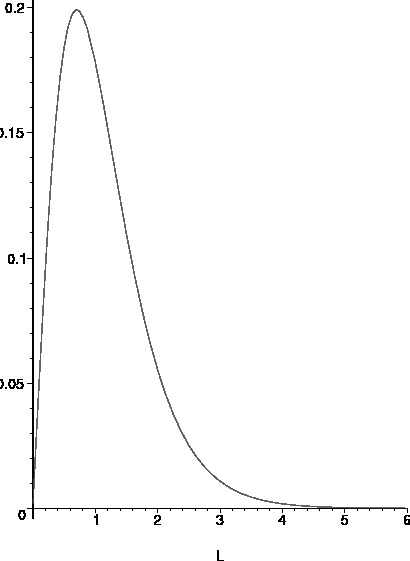}
  \caption{The eigenvalue factor $\Lambda_1$ as function of
  $L$.}
  \label{fig.eigenval_d=1}
\end{figure}
\begin{proof}
  The monotone pinned fluxon for $\gamma=0$ and $d=1+\eps$ with
  $\eps\ll1$ can be written as
\[
\phipin (x;L,0,1+\eps)= \left\{ 
  \begin{arrayl}
  \phi_0(x+\eps x_1^*(L,\eps)),&& x<-L,\\
  \phi_0(x)+\eps\phi_1(x;L,\eps), &&|x|<L,\\
  \phi_0(x-\eps x_1^*(L,\eps)),&& x>L.
  \end{arrayl}
\right.
\]
Here $\phi_1(x;L,\eps)$ is an odd function satisfying
\begin{equation}\label{eq.d=0_phi1}
\eps D_{xx} \phi_1 - (1+\eps)\,\sin(\phi_0+\eps\phi_1)+ \sin\phi_0 =0
,\, |x|<L, 
\end{equation}
and $x_1^*(L,\eps)$ is such that $\phi_0(-L+\eps
x_1^*)=\iphi=\phi_0(-L)+\eps\phi_1(-L)$. To find an approximation for
$\phi_1$ and $x_1^*$, we introduce the notation $\phi^*_0=
\phi_0(-L)$, and $\phi_1^*=\phi_1(-L)$, thus $\eps\phi_1^*=
\iphi-\phi^*_0$. The half length $L$ is
\[
L = \int_{\iphi}^{\pi} \frac{d\phi}{\sqrt{2(h+(1+\eps)(1-\cos\phi)}} =
\int_{\iphi}^{\pi}\frac{d\phi}{\sqrt{2(1-\cos\iphi
    +(1+\eps)(\cos\iphi-\cos\phi))}},
\]
where we used that $h=\eps(\cos\iphi-1)$.  With
$\iphi=\phi_0^*+\eps\phi_1^*$, this becomes
\[
\begin{arrayl}
L&=& \displaystyle\int_{\phi_0^*}^\pi \frac{d\phi}{\sqrt{2(1-\cos\phi)}} 
-\frac\eps2\int _{\phi_0^*}^\pi \
\frac{\cos\phi_0^*-\cos\phi}{(2(1-\cos\phi))^{3/2}} \,d\phi
- \int_{\phi_0^*}^{\phi_0^*+\eps\phi_1^*}
\frac{d\phi}{\sqrt{2(1-\cos\phi)}} +{\mathcal O}(\eps^2)\\[4mm]
&=& \displaystyle L -\frac\eps2\int _{\phi_0^*}^\pi \
\frac{\cos\phi_0^*-\cos\phi}{(2(1-\cos\phi))^{3/2}} \,d\phi
- \eps\int_{0}^{\phi_1^*}
\frac{d\phi}{\sqrt{2(1-\cos\phi^*_0)}} +{\mathcal O}(\eps^2).
\end{arrayl}
\]
Rearranging this expression and using that
$\cos\phi^*_0=1-2\sech^2(L)$, we get an approximation for $\phi_1^*$
\[
\phi_1^* = -\frac{\sech L}8 \,\left[ 2L(1+\tanh^2L)-2\tanh L\right]
+{\mathcal O}(\eps).
\]
Furthermore, $x_1^*$ is given by $\phi_0(-L+\eps
x_1^*)=\phi_0^*+\eps\phi_1^*$. An expansion of $\phi_0(-L+\eps x_1^*)$
shows that
\[
\phi_0^* +\eps x_1^* \phi_0'(-L) = \phi_0^*-\frac{\eps\sech L}8
\,\left[ 2L(1+\tanh^2L)-2\tanh L\right] +{\mathcal O}(\eps^2).
\]
With $\phi_0'(-L) = 2\sech(L)$, this shows that 
\[
x_1^*= -\frac{1}{16}
\,\left[ 2L(1+\tanh^2L)-2\tanh L\right] +{\mathcal O}(\eps).\]

Next, we derive an approximation for the function~$\phi_1$, using the
differential equation~(\ref{eq.d=0_phi1}).
Expanding~(\ref{eq.d=0_phi1}) in~$\eps$ gives
\begin{equation}\label{eq.inhom1}
D_{xx}\phi_1 -\phi_1\cos\phi_0 - \sin\phi_0 = {\mathcal O}(\eps)
\qmbox{or} \L0 \phi_1 = \sin\phi_0 + {\mathcal O}(\eps),
\end{equation}
with $\L0 = D_{xx} - \cos\phi_0$.  The homogeneous problem $\L0\psi=0$
has two independent solutions: $\psi_b (x) = \sech x$ and $\psi_u (x)
= x \, \sech x + \sinh x$. In this, $\psi_b(x) = \frac12 \frac{d}{dx}
\phi_0(x)$ is bounded and $\psi_u(x)$ unbounded as $x \to \pm \infty$.
By the variation-of-constants method, we find the general solution to
(\ref{eq.inhom1}),
\[
\phi_1(x) = 
x\,\sech x + A\,\sech x + B \left[x\,\sech x + \sinh x\right]
+{\mathcal O}(\eps), 
\]
with $A, B \in \mathbb{R}$.  As $\phi_1$ must be odd, it follows that
$A=0$. Furthermore, the boundary condition at $x=-L$ gives
$\phi_1^* = -B\,(L\,\sech L+\sinh L) - L\,\sech L + {\mathcal
  O}(\eps)$, hence
\[
B = \frac{\sech L(L\,\tanh^2 L-\tanh L-3L)}{4(L\,\sech  L+\sinh L)}.  
\]
Altogether we can conclude that $\phi_1(x) = \phi_{11}(x) + {\mathcal
  O}(\eps)$ {with}  
\[
\textstyle
\phi_{11}(x)
= x\,\sech x + \frac{\sech L(L\,\tanh^2 L-\tanh L-3L)}{4(L\,\sech
  L+\sinh L)}\, \left[x\,\sech x + \sinh x\right].
\]
 
To find the largest eigenvalue of $\Lpin(x;h,0,1+\eps)$, we will use
perturbation theory. First we observe that for any $L\geq 0$, the
linearisation $\L0:=\Lpin(x;L,0,1)$ about the fluxon $\phi_0$ has
largest eigenvalue $\Lambda=0$ with eigenfunction is $\phi_0'$.  We
have for $|x|<L$
\[
\Lpin (x;h,0,1+\eps) = D_{xx} - (1+\eps)\cos(\phi_0+\eps\phi_1) = 
\L0(x)   -\eps\, (\cos\phi_0 - \phi_1\,\sin\phi_0) +{\mathcal O}(\eps^2)
\]
and for $x<-L$ 
\[
\Lpin (x;h,0,1+\eps) = \L0(x+\eps x_1^*) = \L0(x) + \eps
x_1^*\phi_0'(x)\sin\phi_0 +{\mathcal O}(\eps^2).
\]
Thus the largest eigenvalue for $\Lpin (x;h,0,1+\eps)$ is
$\Lambda=0+\eps\Lambda_1+{\mathcal O}(\eps^2)$ and the eigenfunction
is $\psi=\phi'_0+\eps \psi_1+{\mathcal O}(\eps^2)$.  The equation for
$\Lambda_1$ and $\psi_1$ is
\begin{equation}\label{eq.perturb_L0}
\L0\psi_1 =  \Lambda_1 \phi'_0 + f_0(x), 
\qmbox{where} f_0(x)= \left\{
\begin{arrayl}
- x_1^*\sin\phi_0(\phi_0')^2 , &&x<-L\\   
\left(\cos\phi_0 - \phi_{11}\sin\phi_0\right)\phi_0', && |x|<L\\
x_1^*\sin\phi_0(\psi_0')^2 , &&x>L
\end{arrayl}\right.
\end{equation}
From~\eqref{eq.inhom1} and the fact that $\L0\phi_0'=0$, it follows that
\[
\begin{array}{ll}
\L0\phi_{11} = \sin\phi_0, &\qmbox{hence} \L0\phi_{11}' = 2\left(\cos\phi_0 -
  \phi_{11}\sin\phi_0\right) \phi_0'\\[2mm] 
\L0\phi_0'=  0 , &\qmbox{hence} \L0\phi_0'' = -\sin\phi_0(\phi'_0)^2.
\end{array}
\]
Thus, 
\[
f_0(x)= \L0\,\left\{
\begin{arrayl}
x_1^*\phi_0''(x), &&x<-L,\\   
\frac12\,\phi_{11}'(x), && |x|<L,\\
-x_1^*\phi_0''(x), &&x>L.
\end{arrayl}\right.
\]
To find~$\Lambda_1$, we multiply the eigenvalue
equation~\eqref{eq.perturb_L0} with $\phi_0'$, integrate it, use
integration by parts and $\L0\phi_0'=0$ and get
\[
\Lambda_1 \int_{-\infty}^\infty (\phi_0')^2\,dx = 
2x_1^*\left[(\phi_0''(L))^2-\phi_0'''(L)\phi_0'(L)\right]
-\phi_{11}''(L)\phi_0'(L)+ \phi_{11}'(L)\phi_0''(L). 
\]
with the explicit expressions for $\phi_0$ and $\phi_1$, we get the
expression in the Lemma.

As the linearisation~$\L0$ about the sine-Gordon fluxon has exactly
one eigenvalue (the one at zero), it follows immediately that if the
perturbed linear operator has more eigenvalues, they have come out of
the continuous spectrum, hence they are near~$-1$.
\end{proof}

\end{document}